\documentclass[runningheads]{llncs}

\usepackage[utf8]{inputenc}
\usepackage{amssymb}
\usepackage{amsmath}
\usepackage{xspace}
\usepackage{hyperref}
\usepackage[capitalize]{cleveref}
\usepackage{xcolor}
\usepackage[inline]{enumitem}

\crefname{subclaim}{Claim}{Claims}
\crefname{property}{Property}{Properties}
\crefname{observation}{Observation}{Observations}

\crefalias{step}{enumi}
\crefname{step}{Step}{Steps}

\crefalias{condition}{enumi}
\crefname{condition}{Condition}{Conditions}

\spnewtheorem{observation}{Observation}{\normalfont\bfseries}{\itshape}

\let\oldproof\proof
\let\oldendproof\endproof
\def\proof{\begingroup \oldproof}
\def\endproof{\qed \oldendproof \endgroup}
\let\origqed\qed
\newcommand{\claimqed}{\hfill$\lrcorner$}
\newenvironment{claimproof}[1][\proofname]{\begin{proof}\renewcommand{\qed}{\claimqed}}{\end{proof}\renewcommand{\qed}{\origqed}}
\newcommand{\qedhere}{} 

\spnewtheorem{subclaim}{Claim}{\itshape}{\rmfamily}

\newcommand{\defparproblem}[5]{
	\vspace{1mm}
	\noindent\fbox{
		\begin{minipage}{0.96\textwidth}
			\begin{tabular*}{\textwidth}{@{\extracolsep{\fill}}lr} \textsc{#1} & {\bf{Parameter:}} #3 \\ \end{tabular*}
			{\bf{Input:}} #2 \\
			{\bf{#5:}} #4
		\end{minipage}
	}
	\vspace{1mm}
}

\DeclareMathOperator{\degree}{deg} 

\DeclareMathOperator{\maxset}{maxset}

\newcommand\restr[2]{{
  \left.\kern-\nulldelimiterspace 
  #1 
  \vphantom{\big|} 
  \right|_{#2} 
  }}

\newcommand{\Oh}{\ensuremath{\mathcal{O}}\xspace}

\newcommand{\cS}{\ensuremath{\mathcal{S}}\xspace}

\newcommand{\T}{\ensuremath{\mathcal{T}}\xspace}
\newcommand{\X}{\ensuremath{\mathcal{X}}\xspace}
\newcommand{\kX}{\ensuremath{\mathfrak{X}}\xspace}

\newcommand{\secl}[3]{\ensuremath{\cS_{#1}^{#2}(#3)}}

\newcommand{\setN}{\ensuremath{\mathbb{N}}\xspace}

\newcommand{\NP}{\ensuremath{\mathsf{NP}}}

\newcommand{\NPhard}{\NP-hard\xspace}

\newcommand{\Wone}{\ensuremath{\mathsf{W[1]}}\xspace}
\newcommand{\Wonehard}{\Wone-hard\xspace}

\newcommand{\LST}{\textsc{LST}}
\newcommand{\ELSS}{\textsc{ELSS}}

\ifdefined\DEBUG{}
\def\rem#1{{\marginpar{\raggedright\scriptsize #1}}}

\newcommand{\jjh}[1]{{\color{orange}{#1}}}
\newcommand{\jjhr}[1]{\rem{\textcolor{orange}{\(\bullet \) #1}}}
\newcommand{\hui}[1]{{\color{red}{#1}}}
\newcommand{\huir}[1]{\rem{\textcolor{red}{\(\bullet \) #1}}}

\newcommand{\bmp}[1]{{\color{blue}{#1}}}
\newcommand{\bmpr}[1]{\rem{\textcolor{blue}{\(\bullet \) #1}}}

\else
\newcommand{\jjh}[1]{#1}
\newcommand{\hui}[1]{#1}
\newcommand{\bmp}[1]{#1}
\newcommand{\jjhr}[1]{}
\newcommand{\huir}[1]{}
\newcommand{\bmpr}[1]{}
\fi

\title{\texorpdfstring{Finding $k$-Secluded Trees Faster}{Finding k-Secluded Trees Faster}}
\author{Huib Donkers\orcidID{0000-0002-2767-8140} \and Bart M.P. Jansen\orcidID{0000-0001-8204-1268}\thanks{Supported by NWO Gravitation grant ``Networks''.} \and Jari J.H. de Kroon\orcidID{0000-0003-3328-9712}}
\institute{Eindhoven University of Technology, Eindhoven, the Netherlands \\
\email{\{h.t.donkers, b.m.p.jansen, j.j.h.d.kroon\}@tue.nl}}

\begin{document}

\maketitle

\begin{abstract}
We revisit the \textsc{$k$-Secluded Tree} problem. Given a vertex-weighted undirected graph~$G$, its objective is to find a maximum-weight induced subtree~$T$ whose open neighborhood has size at most~$k$. We present a fixed-parameter tractable algorithm that solves the problem in time~$2^{\Oh(k \log k)}\cdot n^{\Oh(1)}$, improving on a double-exponential running time from earlier work by Golovach, Heggernes, Lima, and Montealegre. Starting from a single vertex, our algorithm grows a $k$-secluded tree by branching on vertices in the open neighborhood of the current tree~$T$. To bound the branching depth, we prove a structural result that can be used to identify a vertex that belongs to the neighborhood of any $k$-secluded supertree~$T' \supseteq T$ once the open neighborhood of~$T$ becomes sufficiently large. We extend the algorithm to enumerate compact descriptions of all maximum-weight $k$-secluded trees, which allows us to count the number of maximum-weight $k$-secluded trees containing a specified vertex in the same running time.
\end{abstract}

\section{Introduction}\label{secl:sec:intro}

\paragraph{Background} We revisit a problem from the field of parameterized complexity: Given a graph~$G$ with positive weights on the vertices, find a connected induced acyclic subgraph~$H$ of maximum weight such that the open neighborhood of~$H$ in~$G$ has size at most~$k$. Within this field the complexity of a problem is measured not only in the size of the input, but also in terms of some additional parameter. We say that a problem is fixed parameter tractable (FPT) if there is an algorithm that given an instance~$I$ with parameter~$k$, solves the problem in time~$f(k) \cdot |I|^{\Oh(1)}$ for some computable function~$f$. For problems that are FPT, such algorithms allow \NPhard problems to be solved efficiently on instances whose parameter is small. It is therefore desirable for the function~$f$ to grow slowly in terms of~$k$, both out of theoretical interest as well as improving the practical relevance of these algorithms. For a more elaborate introduction to the field we refer to~\cite{DBLP:series/mcs/DowneyF99,DBLP:books/sp/CyganFKLMPPS15}. 

We say that a vertex set~$S\subseteq V(G)$ is $k$-secluded in~$G$ if the open neighborhood of~$S$ in~$G$ has size at most~$k$. An induced subgraph~$H$ of~$G$ is $k$-secluded in~$G$ if~$V(H)$ is. If~$H$ is also a tree, we say that~$H$ is a $k$-secluded tree in~$G$. Formally, the problem we study in this work is defined as follows. 

\defparproblem{Large Secluded Tree (\LST)}{An undirected graph~$G$, a non-negative integer~$k$, and a weight function~$w \colon V(G) \to \mathbb{N}^+$.}{$k$}{Find a $k$-secluded tree~$H$ of~$G$ of maximum weight, or report that no such~$H$ exists.}{Task}

Golovach et al.~\cite{DBLP:journals/jcss/GolovachHLM20} consider the more general \textsc{Connected Secluded $\Pi$-Subgraph}, where the $k$-secluded induced subgraph of~$G$ should belong to some target graph class~$\Pi$. They mention that \textsc{(Large) Secluded Tree} is FPT and can be solved in time~$2^{2^{\Oh(k \log k)}}\cdot n^{\Oh(1)}$ using the recursive understanding technique, the details of which can be found in the arXiv version~\cite{DBLP:journals/corr/GolovachHLM17}. For the case where~$\Pi$ is characterized by a finite set of forbidden induced subgraphs~$\mathcal{F}$, they show that the problem is FPT with a triple-exponential dependency. They pose the question whether it is possible to avoid these double- and triple-exponential dependencies on the parameter. They give some examples of~$\Pi$ for which this is the case, namely for~$\Pi$ being a clique, a star, a $d$-regular graph, or an induced path.

\paragraph{Results}
Our main result is an algorithm for \textsc{Large Secluded Tree} that takes~$2^{\Oh(k \log k)} \cdot n^4$ time. This answers the question of Golovach et al.~\cite{DBLP:journals/jcss/GolovachHLM20} affirmatively for the case of trees. We solve a more general version of the problem, where a set of vertices is given that should be part of the $k$-secluded tree. Our algorithm goes one step further by allowing us to find all maximum weight solutions. As we will later argue, it is not possible to output all such solutions directly in the promised running time. Instead, the output consists of a bounded number of solution descriptions such that each maximum weight solution can be constructed from one such description. This is similar in spirit to the work of Guo et al.~\cite{DBLP:journals/jcss/GuoGHNW06}, who enumerate all minimal solutions to the \textsc{Feedback Vertex Set} problem in~$\Oh(c^k \cdot m)$ time. They do so by giving a list of \emph{compact representations}, a set~$\mathcal{C}$ of pairwise disjoint vertex subsets such that choosing exactly one vertex from every set results in a minimal feedback vertex set. Our descriptions allow us to \emph{count} the number of maximum-weight $k$-secluded trees containing a specified vertex in the same running time.

\paragraph{Techniques}
Rather than using recursive understanding, our algorithm is based on bounded-depth branching with a non-trivial progress measure. Similarly to existing algorithms to compute spanning trees with many leaves~\cite{DBLP:journals/algorithmica/KneisLR11}, our algorithm iteratively grows the vertex set of a $k$-secluded tree~$T$. If we select a vertex~$v$ in the neighborhood of the current tree~$T$, then for any $k$-secluded supertree~$T'$ of~$T$ there are two possibilities: either~$v$ belongs to the neighborhood of~$T'$, or it is contained in~$T'$; the latter case can only happen if~$v$ has exactly one neighbor in~$T$. Solutions of the first kind can be found by deleting~$v$ from the graph and searching for a $(k-1)$-secluded supertree of~$T$. To find solutions of the second kind we can include~$v$ in~$T$, but since the parameter does not decrease in this case we have to be careful that the recursion depth stays bounded. Using a reduction rule to deal with degree-1 vertices, we can essentially ensure that~$v$ has at least three neighbors (exactly one of which belongs to~$T$), so that adding~$v$ to~$T$ strictly increases the open neighborhood size~$|N(T)|$. Our main insight to obtain an FPT algorithm is a structural lemma showing that, whenever~$|N(T)|$ becomes sufficiently large in terms of~$k$, we can identify a vertex~$u$ that belongs to the open neighborhood of any $k$-secluded supertree~$T' \supseteq T$. At that point, we can remove~$u$ and decrease~$k$ to make progress.

\paragraph{Related work}
Secluded versions of several classic optimization problems have been studied intensively in recent years~\cite{DBLP:journals/disopt/BevernFMMSS18,DBLP:journals/networks/BevernFT20,DBLP:journals/algorithmica/ChechikJPP17,DBLP:journals/mst/FominGKK17,DBLP:journals/ipl/LuckowF20}, many of which are discussed in Till Fluschnik's PhD thesis~\cite{Fluschnik2020}. Marx~\cite{DBLP:journals/tcs/Marx06} considers a related problem \textsc{Cutting~$k$ (connected) vertices}, where the aim is to find a (connected) set~$S$ of size exactly~$k$ with at most~$\ell$ neighbors. Without the connectivity requirement, the problem is \Wonehard by~$k+\ell$. The problem becomes FPT when~$S$ is required to be connected, but remains \Wonehard by~$k$ and~$\ell$ separately. Fomin et al.~\cite{DBLP:conf/mfcs/FominGK13} consider the variant where~$|S| \leq k$ and show that it is FPT parameterized by~$\ell$.

\paragraph{Organization}

We introduce our enumeration framework in \cref{sec:prelims}. We present our algorithm that enumerates maximum-weight $k$-secluded trees in \cref{sec:enum_supertrees} and present its correctness and running time analyses. We give some conclusions in \cref{sec:conclusion}.

\section{Framework for enumerating secluded trees}\label{sec:prelims}

We consider simple undirected graphs with vertex set~$V(G)$ and edge set~$E(G)$. We use standard notation pertaining to graph algorithms, such as presented by Cygan et al.~\cite{DBLP:books/sp/CyganFKLMPPS15}. When the graph~$G$ is clear from context, we denote~$|V(G)|$ and~$|E(G)|$ by~$n$ and~$m$ respectively. For an induced subgraph~$H$ of~$G$, we may write~$N(H)$ to denote~$N(V(H))$. If~$w \colon V(G) \to \setN^+$ is a weight function, then for any~$S \subseteq V(G)$ let~$w(S) := \sum_{v \in S} w(s)$ and for any subgraph~$H$ of~$G$ we may denote~$w(V(H))$ by~$w(H)$. 

It is not possible to enumerate all maximum-weight $k$-secluded trees in FPT time; consider the graph with~$n$ vertices of weight~$1$ and two vertices of weight~$n$ which are connected by~$k+1$ vertex-disjoint paths on~$n/(k+1)$ vertices each, then there are~$\Oh(k \cdot (n/k)^k)$ maximum-weight $k$-secluded trees which consist of all vertices except one vertex out of exactly~$k$ paths. However, it is possible to give one short description for such an exponential number of $k$-secluded trees.

\begin{definition}\label{def:description}

 For a graph~$G$, a \emph{description} is a pair~$(r,\mathcal{X})$ consisting of a vertex~$r \in V(G)$ and a set~$\mathcal{X}$ of pairwise disjoint subsets of~$V(G-r)$ such that for any set~$S$ consisting of exactly one vertex from each set~$X \in \mathcal{X}$, the connected component~$H$ of~$G-S$ containing~$r$ is acyclic and~$N(H) = S$, i.e.,~$H$ is a $|\mathcal{X}|$-secluded tree in~$G$. The \emph{order} of a description is equal to~$|\mathcal{X}|$.
 We say that a $k$-secluded tree~$H$ is \emph{described by} a description~$(r,\mathcal{X})$ if~$N(H)$ consists of exactly one vertex of each~$X \in \mathcal{X}$ and~$r \in V(H)$.
\end{definition}

\begin{definition}
 For a graph~$G$, a set of descriptions~\kX of maximum order~$k$ is called \emph{redundant} for~$G$ if there is a $k$-secluded tree~$H$ in~$G$ such that~$H$ is described by two distinct descriptions in~\kX. We say~\kX is \emph{non-redundant} for~$G$ otherwise. 
\end{definition}

\begin{definition}
 For a graph~$G$ and a set of descriptions~\kX of maximum order~$k$, let~$\T_G(\kX)$ denote the set of all $k$-secluded trees in~$G$ described by a description in~\kX.
\end{definition}

\begin{observation} \label{obs:descr-union}
 For a graph~$G$ and two sets of descriptions~$\kX_1, \kX_2$ we have:~$$\T_G(\kX_1) \cup \T_G(\kX_2) = \T_G(\kX_1 \cup \kX_2).$$
\end{observation}

\begin{observation} \label{obs:descr-merge}
For a graph~$G$, a set of descriptions~\kX, and vertex sets~$X_1, X_2$ disjoint from~$\bigcup_{(r,\X) \in \kX} (\{r\} \cup \bigcup_{X \in \X} X)$, the set~$\T_G(\{(r,\X \cup \{X_1 \cup X_2\}) \mid (r,\X) \in \kX)\}$ equals:~$$\T_G(\{(r,\X \cup \{X_1\}) \mid (r,\X) \in \kX)\} \cup (r,\X \cup \{X_2\}) \mid (r,\X) \in \kX)\}).$$
\end{observation}

For an induced subgraph~$H$ of~$G$ and a set~$F \subseteq V(G)$, we say that~$H$ is a supertree of~$F$ if~$H$ induces a tree and~$F \subseteq V(H)$. 
Let~$\secl{G}{k}{F}$ be the set of all $k$-secluded supertrees of~$F$ in~$G$. For a set~$X$ of subgraphs of~$G$ let~$\maxset_w(X) := \{H \in X \mid w(H) \geq w(H') \text{~for all~} H' \in X\}$. We focus our attention to the following version of the problem, where some parts of the tree are already given. 

\defparproblem{Enumerate Large Secluded Supertrees (\ELSS)}{A graph~$G$, a non-negative integer~$k$, non-empty vertex sets~$T \subseteq F \subseteq V(G)$ such that~$G[T]$ is connected, and a weight function~$w \colon V(G) \to \mathbb{N}^+ $.}{$k$}{A non-redundant set~$\mathfrak{X}$ of descriptions such that~$\mathcal{T}_G(\mathfrak{X}) = \maxset_w(\secl{G}{k}{F})$.
}{Output}

In the end we solve the general enumeration problem by solving \ELSS{} with~$F = T = \{v\}$ for each~$v \in V(G)$ and reporting only those $k$-secluded trees of maximum weight.
Intuitively, our algorithm for \ELSS{} finds $k$-secluded trees that ``grow'' out of~$T$. In order to derive some properties of the types of descriptions we compute, we may at certain points demand that certain vertices non-adjacent to~$T$ need to end up in the $k$-secluded tree. For this reason the input additionally has a set~$F$, rather than just~$T$.

Our algorithm solves smaller instances recursively. We use the following abuse of notation: in an instance with graph~$G$ and weight function~$w \colon V(G) \to \mathbb{N}^+$, when solving the problem recursively for an instance with induced subgraph~$G'$ of~$G$, we keep the weight function~$w$ instead of restricting the domain of~$w$ to~$V(G')$.

\begin{observation} \label{obs:secluded}
 For a graph~$G$, a vertex~$v \in V(G)$, and an integer~$k \geq 1$, if~$H$ is a $(k-1)$-secluded tree in~$G-v$, then~$H$ is a $k$-secluded tree in~$G$. Consequently,~$\secl{G-v}{k-1}{F} \subseteq \secl{G}{k}{F}$ for any~$F \subseteq V(G)$.
\end{observation}

\begin{observation} \label{obs:secluded2}
 For a graph~$G$, a vertex~$v \in V(G)$, and an integer~$k \geq 1$, if~$H$ is a $k$-secluded tree in~$G$ with~$v \in N_G(H)$, then~$H$ is a $(k-1)$-secluded tree in~$G-v$. Consequently,~$\{H \in \secl{G}{k}{F} \mid v \in N_G(H)\} \subseteq \secl{G-v}{k-1}{F}$ for any~$F \subseteq V(G)$.
\end{observation}

\section{Enumerate large secluded supertrees}\label{sec:enum_supertrees}

\Cref{sec:subroutines} proves the correctness of a few subroutines used by the algorithm. \Cref{sec:alg} describes the algorithm to solve \ELSS. In \cref{sec:poc} we prove its correctness and in \cref{sec:time} we analyze its time complexity. In \cref{sec:enum_count} we show how the algorithm for \ELSS{} can be used to count and enumerate maximum-weight $k$-secluded trees containing a specified vertex.

\subsection{Subroutines for the algorithm}\label{sec:subroutines}

Similar to the \textsc{Feedback Vertex Set} algorithm given by Guo et al.~\cite{DBLP:journals/jcss/GuoGHNW06}, we aim to get rid of degree-1 vertices. In our setting there is one edge case however. The reduction rule is formalized as follows.


\begin{definition}\label{def:contract}
 For an \ELSS{} instance~$(G,k,F,T,w)$ with a degree-1 vertex~$v$ in~$G$ such that~$F \neq \{v\}$, \emph{contracting~$v$ into} its neighbor~$u$ yields the \ELSS{} instance~$(G-v,k,F',T',w')$ where the weight of~$u$ is increased by~$w(v)$ and:
$$\begin{array}{ccc}
F' = \begin{cases}
(F \setminus \{v\}) \cup \{u\} & \mbox{if~$v \in F$} \\
F & \mbox{otherwise}
\end{cases} & \quad\quad &
T' = \begin{cases}
(T \setminus \{v\}) \cup \{u\} & \mbox{if~$v \in T$} \\
T & \mbox{otherwise.}
\end{cases}
\end{array}
$$ 
\end{definition}
We prove the correctness of the reduction rule, that is, the descriptions of the reduced instance form the desired output for the original instance.

\newcommand{\lemSafeContract}{
Let~$I = (G,k,F,T,w)$ be an \ELSS{} instance. Suppose~$G$ contains a degree-1 vertex~$v$ such that~$\{v\} \neq F$. Let~$I' = (G-v,k,F',T',w')$ be the instance obtained by contracting~$v$ into its neighbor~$u$. If~$\kX$ is a non-redundant set of descriptions for~$G-v$ such that~$\T_{G-v}(\kX) = \maxset_{w'}(\secl{G-v}{k}{F'})$, 
then~$\kX$ is a non-redundant set of descriptions for~$G$ such that~$\T_G(\kX) = \maxset_w(\secl{G}{k}{F})$. 
}
\begin{lemma} \label{lem:safe:contract}
\lemSafeContract
\end{lemma}
\begin{proof}
We first argue that~$\kX$ is a valid set of descriptions for the graph~$G$. For any~$(r,\X) \in \kX$, we have~$r  \in V(G - v)$ and~$\X$ consists of disjoint subsets of~$V(G-\{v,r\})$, which trivially implies that~$r \in V(G)$ and that~$\X$ consists of disjoint subsets of~$V(G-r)$. Consider any set~$S$ consisting of exactly one vertex from each~$X \in \X$. The connected component~$H$ of~$(G-v)-S$ containing~$r$ is acyclic and~$N_{G-v}(H) = S$ since \kX is a description for~$G-v$. 
Let~$H'$ be the connected component of~$G-S$ containing~$r$. Note that~$V(H) \subseteq V(H') \subseteq V(H) \cup \{v\}$. Clearly~$H'$ is acyclic as it is obtained from~$H$ by possibly adding a degree-1 vertex. We argue that~$N_G(H') = S$. By construction of~$H'$ we have~$N_G(H') \subseteq S$. For the sake of contradiction suppose that there is some~$p \in S \setminus N_G(H')$. Since~$N_{G-v}(H) = S$, there is some vertex~$q \in V(H)$ such that~$p \in N_{G-v}(q)$. But since~$q \in V(H')$, we have that~$p \in N_G(q)$ and so~$p \in N_G(H')$; a contradiction to the containment of~$p \in S \setminus N_G(H')$. It follows that~$N_G(H') = S$.

Observe that any maximum weight $k$-secluded supertree~$H$ of~$F$ in~$G$ containing~$u$, contains its neighbor~$v$ as well: adding~$v$ to an induced tree subgraph containing~$u$ does not introduce cycles since~$v$ has degree one, does not increase the size of the neighborhood, and strictly increases the weight since~$w(v) > 0$ by definition. Conversely, any (maximum weight) $k$-secluded supertree~$H$ of~$F$ in~$G$ that contains~$v$ also contains~$u$: tree~$H$ contains all vertices of the non-empty set~$F$ and~$F \neq \{v\}$, so~$H$ contains at least one vertex other than~$v$, which implies by connectivity that it contains the unique neighbor~$u$ of~$v$. Hence a maximum weight $k$-secluded supertree~$H$ of~$F$ in~$G$ contains~$u$ if and only if it contains~$v$.

Using this fact, we relate the sets~$\maxset_{w}(\secl{G}{k}{F})$ and~$\maxset_{w'}(\secl{G-v}{k}{F'})$. For any~$H \in \maxset_{w}(\secl{G}{k}{F})$, there is a $k$-secluded supertree of~$F'$ in~$G-v$ of the same weight under~$w'$:
\begin{itemize}
    \item If~$v \in H$, then~$u \in H$ as argued above. Now observe that~$H-v$ is a $k$-secluded tree in~$G-v$ that contains~$u$. Since the weight of~$u$ was increased by~$w(v)$ in the transformation, we have~$w(H) = w'(H-v)$. Since~$H$ is a supertree of~$F$ and~$H-v$ contains~$u$, the latter is a supertree of~$F' \subseteq (F \setminus \{v\}) \cup \{u\}$.
    \item If~$v \notin H$, then~$u \notin H$ and therefore~$w(H) = w'(H)$ and~$N_G(H) = N_{G-v}(H)$. As~$v\notin H$ while~$H$ is a supertree of~$F \supseteq T$ we have~$v \notin F \cup T$, which shows~$F' = F$ and~$T' = T$ so that~$H$ is a $k$-secluded supertree of~$F'$ in~$G-v$.
\end{itemize}

Conversely, for any $k$-secluded supertree~$H'$ of~$F'$ in~$G-v$, there is a $k$-secluded supertree of~$F$ in~$G$ of the same weight under~$w$: if~$u \notin H'$ then~$H'$ itself is such a tree, otherwise~$G[V(H') \cup \{v\}]$ is such a tree. 

These transformations imply that the maximum $w$-weight of trees in~$\secl{G}{k}{F}$ is identical to the maximum $w'$-weight of trees in~$\secl{G-v}{k}{F'}$. \bmp{Since any~$H \in \maxset_{w}(\secl{G}{k}{F})$ contains~$u$ if and only if it contains~$v$, they also show that an induced subgraph~$H$ of~$G$ containing either both~$\{u,v\}$ or none belongs to~$\maxset_{w}(\secl{G}{k}{F})$ if and only if~$H - v \in \maxset_{w'}(\secl{G-v}{k}{F'})$;} note that this holds regardless of whether~$v \in H$. To show that~$\T_G(\kX) = \maxset_w(\secl{G}{k}{F})$ it now suffices to observe that if~$(r,\X) \in \kX$ describes a tree~$H - v \in \maxset_{w'}(\secl{G-v}{k}{F'})$ via the set~$S$ containing exactly one vertex of each~$X \in \mathcal{X}$, such that~$H - v$ is the connected component of~$(G - v) - S$ containing~$r$, then the connected component of~$G - S$ containing~$r$ is exactly~$H$, which again holds regardless of whether~$v \in H$. Hence~$\T_G(\kX) = \maxset_w(\secl{G}{k}{F})$, and the set of descriptions is non-redundant for~$G$ since~$\kX$ is non-redundant for~$G-v$.
\end{proof}

We say an instance is \emph{almost leafless} if the lemma above cannot be applied, that is, if~$G$ contains a vertex~$v$ of degree 1, then~$F=\{v\}$.

\begin{lemma} \label{lem:greedy:naive}
    There is an algorithm that, given an almost leafless \ELSS{} instance~$(G,k,F,T,w)$ such that~$k>0$ and~$|N_G(T)| > k(k+1)$, runs in time~$\Oh(k \cdot n^3)$ and either:
    \begin{enumerate}
        \item finds a vertex~$v \in V(G)\setminus F$ such that any $k$-secluded supertree~$H$ of~$F$ in~$G$ satisfies~$v \in N_G(H)$, or
        \item concludes that~$G$ does not contain a $k$-secluded supertree of~$F$.
    \end{enumerate}
\end{lemma}
\begin{proof}
%
    We aim to find a vertex~$v \in V(G) \setminus F$ with~$k+2$ distinct paths~$P_1,\dots,P_{k+2}$ from~$N_G(T)$ to~$v$ that intersect only in~$v$ and do not contain vertices from~$T$. We first argue that such a vertex~$v$ satisfies the first condition, if it exists. Consider some $k$-secluded supertree~$H$ of~$F$. Since the paths~$P_1, \dots, P_{k+2}$ are disjoint apart from their common endpoint~$v$ while~$|N_G(H)| \leq k$, there are two paths~$P_i, P_j$ with~$i \neq j \in [k+2]$ for which~$P_i \setminus \{v\}$ and~$P_j \setminus \{v\}$ do not intersect~$N_G(H)$. Since they start in~$N_G(T)$, the paths~$P_i \setminus \{v\}$ and~$P_j \setminus \{v\}$ are contained in~$H$. As~$P_i$ and~$P_j$ form a cycle together with a path through the connected set~$T$, which cannot be contained in the acyclic graph~$H$, this implies~$v \in N_G(H)$.

    Next we argue that if~$G$ has a $k$-secluded supertree~$H$ of~$F \supseteq T$, then there exists such a vertex~$v$. Consider an arbitrary such~$H$ and root it at a vertex~$t \in T$. For each vertex~$u \in N_G(T)$, we construct a path~$P_u$ disjoint from~$T$ that starts in~$u$ and ends in~$N_G(H)$, as follows.
    \begin{itemize}
        \item If~$u \notin H$, then~$u \in N_G(H)$ and we take~$P_u = (u)$.
        \item If~$u \in H$, then let~$\ell_u$ be an arbitrary leaf in the subtree of~$H$ rooted at~$u$; possibly~$u = \ell_u$. Since~$T$ is connected and~$H \supseteq T$ is acyclic and rooted in~$t \in T$, the subtree rooted at~$u \in N_G(T) \cap H$ is disjoint from~$T$. Hence~$\ell_u \notin T$, so that~$F \neq \{\ell_u\}$. As the instance is almost leafless we therefore have~$\deg_G(\ell_u) > 1$. Because~$\ell_u$ is a leaf of~$H$ this implies that~$N_G(\ell_u)$ contains a vertex~$y$ other than the parent of~$\ell_u$ in~$H$, so that~$y \in N_G(H)$. We let~$P_u$ be the path from~$u$ to~$\ell_u$ through~$H$, followed by the vertex~$y \in N_G(H)$.
    \end{itemize}
    The paths we construct are distinct since their startpoints are. Two constructed paths cannot intersect in any vertex other than their endpoints, since they were extracted from different subtrees of~$H$. Since we construct~$|N_G(T)| > k(k+1)$ paths, each of which ends in~$N_G(H)$ which has size at most~$k$, some vertex~$v \in N_G(H)$ is the endpoint of~$k+2$ of the constructed paths. As shown in the beginning the proof, this establishes that~$v$ belongs to the neighborhood of any $k$-secluded supertree of~$F$. Since~$F \subseteq V(H)$ we have~$v \notin F$.

    All that is left to show is that we can find such a vertex~$v$ in the promised time bound. After contracting~$T$ into a source vertex~$s$, for each~$v \in V(G) \setminus F$, do~$k+2$ iterations of the Ford-Fulkerson algorithm in order to check if there are~$k+2$ internally vertex-disjoint $sv$-paths. If so, then return~$v$. If for none of the choices of~$v$ this holds, then output that there is no $k$-secluded supertree of~$F$ in~$G$. 
    In order to see that this satisfies the claimed running time bound, note that there are~$\Oh(n)$ choices for~$v$, and~$k+2$ iterations of Ford-Fulkerson runs can be implemented to run in~$\Oh(k\cdot(n+m))$ time.
\end{proof}

\subsection{The algorithm} \label{sec:alg}

Consider an input instance~$(G,k,F,T,w)$ of \ELSS. If~$G[F]$ contains a cycle, return~$\emptyset$. Otherwise we remove all connected components of~$G$ that do not contain a vertex of~$F$. If more than one connected component remains, return~$\emptyset$.
Then, while there is a degree-1 vertex~$v$ such that~$F \neq \{v\}$, contract~$v$ into its neighbor as per~\cref{def:contract}.
While~$N_G(T)$ contains a vertex~$v \in F$, add~$v$ to~$T$.
Finally, if~$N_G(F) = \emptyset$, return~$\{(r,\emptyset)\}$ for some~$r \in F$.
Otherwise if~$k=0$, return~$\emptyset$.



We proceed by considering the neighborhood of~$T$ as follows:
\begin{enumerate}

 \item \label[step]{alg:2nbrs} If any vertex~$v \in N_G(T)$ has two neighbors in~$T$, then recursively run this algorithm to obtain a set of descriptions~$\mathfrak{X}'$ for~$(G-v,k-1,F,T,w)$ and return~$\{(r,\mathcal{X} \cup \{\{v\}\}) \mid (r,\mathcal{X}) \in \mathfrak{X}'\}$.
 

 \item \label[step]{alg:greedy} If~$|N_G(T)| > k(k+1)$, apply Lemma~\ref{lem:greedy:naive}. If it concludes that~$G$ does not contain a $k$-secluded supertree of~$F$, return~$\emptyset$. Otherwise let~$v \in V(G) \setminus F$ be the vertex it finds, obtain a set of descriptions~$\mathfrak{X}'$ for~$(G-v, k-1, F, T, w)$ and return~$\{(r,\mathcal{X} \cup \{\{v\}\}) \mid (r,\mathcal{X}) \in \mathfrak{X}'\}$.

 \item \label[step]{alg:3way} Pick some~$v \in N_G(T)$ and let~$P = (v=v_1,v_2,\ldots,v_\ell)$ be the unique\footnote{
 To construct~$P$, initialize~$P := (v=v_1)$; then while~$\deg_G(v_{|V(P)|}) = 2$ and~$N_G(v_{|V(P)|}) \setminus (V(P) \cup T)$ consists of a single vertex, append that vertex to~$P$.
 \label{footnote}} maximal path disjoint from~$T$ satisfying~$\degree_G(v_i) = 2$ for each~$1 \leq i < \ell$ and~($v_\ell \in N_G(T)$ or~$\degree_G(v_\ell) > 2$).
 \begin{enumerate}
     \item If~$v_\ell \not\in F$, obtain a set of descriptions~$\mathfrak{X}_1$ by recursively solving~$(G-v_\ell, k-1, F,T, w)$. Otherwise take~$\mathfrak{X}_1 = \emptyset$. (We find the $k$-secluded trees avoiding~$v_\ell$ but containing~$P - v_\ell$.)
     
     \item If~$P - F - v_\ell \neq \emptyset$, obtain a set of descriptions~$\mathfrak{X}_2$ by recursively solving~$(G-V(P-v_\ell), k-1, (F \setminus V(P)) \cup \{v_\ell\}, T, w)$. Otherwise take~$\mathfrak{X}_2 = \emptyset$. \jjhr{I believe this should be fine, the updated~$T'$ seems only to be used to keep it a connected component (whose definition was wrong I think)}
     (We find the $k$-secluded trees containing both endpoints of~$P$ which have one vertex in~$P$ as a neighbor.)
     
     \item If~$G[F \cup V(P)]$ is acyclic, obtain a set of descriptions~$\mathfrak{X}_3$ by recursively solving~$(G,k,F \cup V(P), T \cup V(P), w)$. Otherwise take~$\mathfrak{X}_3 = \emptyset$. (We find the $k$-secluded trees containing the entire path~$P$.)
 \end{enumerate}
 
 %
 Let~$M$ be the set of minimum weight vertices in~$P-F-v_\ell$ and define:
 \begin{align*}
     \kX'_1 &:= \{(r,\mathcal{X} \cup \{\{v_\ell\}\}) \mid (r,\mathcal{X}) \in \mathfrak{X}_1\} \\
     \kX'_2 &:= \{(r,\mathcal{X} \cup \{M\}) \mid (r,\mathcal{X}) \in \mathfrak{X}_2\} \\
     \kX'_3 &:= \kX_3. 
 \end{align*}
 For each~$i \in [3]$ let~$w_i$ be the weight of an arbitrary~$H \in \T_G(\kX'_i)$, or~$0$ if~$\kX'_i = \emptyset$. Return the set~$\kX'$ defined as~$\bigcup_{\{i \in [3] \mid w_i=\max\{w_1,w_2,w_3\}\}} \kX'_i$.
\end{enumerate}

\subsection{Proof of correctness} \label{sec:poc}

In this section we argue that the algorithm described in \cref{sec:alg} solves the \ELSS{} problem. In various steps we identify a vertex~$v$ such that the neighborhood of any maximum-weight $k$-secluded supertree must include~$v$. We argue that for these steps, the descriptions of the current instance can be found by adding~$\{v\}$ to every description of the supertrees of~$T$ in~$G-v$ if some preconditions are satisfied.

\newcommand{\lemCombine}{
Let~$(G,k,F,T,w)$ be an \ELSS{} instance and let~$v \in V(G) \setminus F$. Let~\kX be a set of descriptions for~$G-v$ such that~$\T_{G-v}(\kX) = \maxset_w(\secl{G-v}{k-1}{F})$
and~$v \in N_G(H)$ for all~$H \in \T_{G-v}(\kX)$. Then we have:
$$\T_G \left(\{(r, \X \cup \{\{v\}\}) \mid (r, \X) \in \kX\} \right) = \maxset_w \{H \in \secl{G}{k}{F} \mid v \in N_G(H)\}.$$ 


}
\begin{lemma} \label{lem:combine}
\lemCombine
\end{lemma}
\begin{proof}
First observe that~$\kX' = \{(r, \X \cup \{\{v\}\}) \mid (r, \X) \in \kX\}$ is a valid set of descriptions for~$G$ since~$v \in N_G(H)$ for all~$H \in \T_{G-v}(\kX)$. 

Consider a maximum-weight (with respect to~$w$) $k$-secluded supertree~$H \in \secl{G}{k}{F}$ such that~$v \in N_G(H)$. We show that it is contained in~$\T_G(\kX')$. By \cref{obs:secluded2} we have that~$H \in \secl{G-v}{k-1}{F}$, that is,~$H$ is a $(k-1)$-secluded supertree of~$F$ in~$G-v$. We argue that~$H \in \maxset_w(\secl{G-v}{k-1}{F})$. For the sake of contradiction, suppose there is~$H' \in \maxset_w(\secl{G-v}{k-1}{F})$ such that~$w(H') > w(H)$.
By \cref{obs:secluded} it follows that~$H'$ is a $k$-secluded supertree of~$F$ in~$G$. This contradicts the fact that~$H$ is maximum weight among such supertrees. It follows that~$H \in \maxset_w(\secl{G-v}{k-1}{F})$. By the assumption that~$\T_{G-v}(\kX) = \maxset_w(\secl{G-v}{k-1}{F})$, we have that there is a description~$(r,\X) \in \kX$ for~$G-v$ that describes~$H$. Since~$(r, \X \cup \{\{v\}\}) \in \kX'$ is a description for~$H$ in~$G$, we have that~$H \in \T_G(\kX')$ as required.

In the other direction, consider some tree~$J \in \T_G(\kX')$. We show that~$J \in \maxset_w \{H \in \secl{G}{k}{F} \mid v \in N_G(H)\}$. Let~$(r,\X \cup \bmp{\{\{v\}\}}) \in \kX'$ be a description that describes~$J$. By \cref{def:description} we have~$v \in N_G(J)$. Since~$(r, \X) \in \kX$ describes~$J$ in~$G-v$ and~$\T_{G-v}(\kX) = \maxset_w(\secl{G-v}{k-1}{F})$, we have that~$J \in \maxset_w(\secl{G-v}{k-1}{F})$. Since~$v \in N_G(J)$, by \cref{obs:secluded} we have that~$J \in \{H \in \secl{G}{k}{F} \mid v \in N_G(H)\}$. For the sake of contradiction, suppose that there is~$J' \in \{H \in \secl{G}{k}{F} \mid v \in N_G(H)\}$ such that~$w(J') > w(J)$. Then we get that~$J' \in \maxset_w(\secl{G-v}{k-1}{F})$, but this contradicts that~$J \in \maxset_w(\secl{G-v}{k-1}{F})$. It follows that~$J \in \maxset_w \{H \in \secl{G}{k}{F} \mid v \in N_G(H)\}$ as required.
\end{proof}

\bmp{The next lemma will be used to argue correctness of \cref{alg:3way}(b) of the algorithm, in which we find $k$-secluded trees which avoid a single vertex from path~$P$.}

\begin{lemma} \label{lem:combine-path}
Let~$(G,k,F,T,w)$ be an \ELSS{} instance and let~$P$ be a path in~$G$ with~$\degree_G(v) = 2$ for all~$v \in V(P)$ and~$N_G(P) = \{a,b\}$ for some~$a,b \in F$. Let~\kX be a set of descriptions for~$G-V(P)$ such that~$\T_{G-V(P)}(\kX) = \maxset_w(\secl{G-V(P)}{k-1}{F \setminus V(P)})$. Then for all~$p \in V(P) \setminus F$ we have: 
$$\T_G(\{(r, \X \cup \{\{p\}\}) \mid (r, \X) \in \kX\}) = \maxset_w\{H \in \secl{G}{k}{F} \mid p \in N_G(H)\}.$$
\end{lemma}
\begin{proof}
\hui{

For any~$p \in V(P)$ we define~$\kX^p = \{(r, \X \cup \{\{p\}\}) \mid (r, \X) \in \kX\}$. We show~$\kX^p$ is a valid set of descriptions for~$G$. Note that for any set~$S$ consisting of exactly one vertex from each set~$X \in \X \cup \{\{p\}\}$ there is an~$H \in T_{G-V(P)}(\kX)$ with~$N_G(H) = S\setminus\{p\}$ and~$r \in V(H)$ since~$\kX$ is a valid set of descriptions for~$G-V(P)$. Since~$\{a,b\} \subseteq V(H)$ and~$N_G(P) = \{a,b\}$, we have that~$H' = G[V(H) \cup V(P)]$ is the connected component of~$G-(S \setminus p)$ containing~$r$. Observe that~$H'-p$ is acyclic and connected and since~$p \in V(P) \subseteq V(H')$ we have that~$p \in N_G(H'-p)$, hence~$N_G(H'-p) = S$ and~$\kX^p$ is a valid set of descriptions for~$G$.

Next we show~$\T_G(\kX^p) \supseteq \maxset_w\{H \in \secl{G}{k}{F} \mid p \in N_G(H)\}$ for any~$p \in V(P) \setminus F$.
\bmp{Consider a $k$-secluded supertree~$H \in \secl{G}{k}{F}$ which has maximum weight (with respect to~$w$) among those satisfying~$p \in N_G(H)$.}
We show that~$H \in \T_G(\kX^p)$. By Note~\ref{obs:secluded2} we have that~$H \in \secl{G-p}{k-1}{F}$, that is,~$H$ is a $(k-1)$-secluded supertree of~$F$ in~$G-p$. Observe that~$H-V(P)$ remains connected so~$H-V(P) \in \secl{G-V(P)}{k-1}{F \setminus V(P)}$. We argue that~$H-V(P) \in \maxset_w(\secl{G-V(P)}{k-1}{F \setminus V(P)})$. For the sake of contradiction, suppose there is~$H' \in \maxset_w(\secl{G-V(P)}{k-1}{F \setminus V(P)})$ such that~$w(H') > w(H-V(P))$. Observe that~$H'' := G[V(H') \cup V(P-p)]$ is a connected acyclic subgraph of~$G$ with~$N_G(H'') = N_G(H') \cup \{p\}$, i.e.,~$H''$ is a $k$-secluded supertree of~$F$ in~$G$. Since~$w(H'') = w(H') + w(P-p) > w(H-V(P)) + w(P-p) = w(H)$ this contradicts that~$H \in \maxset_w(\secl{G}{k}{F})$. It follows that~$H-V(P) \in \maxset_w(\secl{G-V(P)}{k-1}{F \setminus V(P)})$. Since it is given that~$\T_{G-V(P)}(\kX) = \maxset_w(\secl{G-V(P)}{k-1}{F \setminus V(P)})$, we have that there is a description~$(r,\X) \in \kX$ for~$G-V(P)$ that describes~$H-V(P)$. Then~$(r, \X \cup \{\{p\}\}) \in \kX^p$ is a description for~$H$ in~$G$ and we conclude that~$H \in \T_G(\kX^p)$ as required.

Finally we show~$\T_G(\kX^p) \subseteq \maxset_w\{H \in \secl{G}{k}{F} \mid p \in N_G(H)\}$ for any~$p \in V(P) \setminus F$.
Consider some tree~$J \in \T_G(\kX^p)$. We show that~$J \in \maxset_w \{H \in \secl{G}{k}{F} \mid p \in N_G(H)\}$.
Clearly~$J \in \{H \in \secl{G}{k}{F} \mid p \in N_G(H)\}$, so it remains to show that~$w(J) \geq w(J')$ for all~$J' \in \{H \in \secl{G}{k}{F} \mid p \in N_G(H)\}$. Suppose for contradiction that there exists such a~$J'$ for which~$w(J) < w(J')$. Observe that~$J - V(P)$ and~$J' - V(P)$ are both $(k-1)$-secluded supertrees of~$F \setminus V(P)$ in~$G-V(P)$, i.e., they are contained in~$\secl{G-V(P)}{k-1}{F \setminus V(P)}$. Since~$V(P-p) \subseteq V(J)$ and~$V(P-p) \subseteq V(J')$ we have that~$w(J-V(P)) < w(J'-V(P))$, so~$J-V(P) \not\in \maxset_w(\secl{G-V(P)}{k-1}{F \setminus V(P)})$. Recall that~$J \in T_G(\kX^p)$ and consider the description~$(r,\X \cup \{\{p\}\}) \in \kX^p$ that describes~$J$. Observe that~$(r,\X)$ describes~$J-V(P)$ in~$G-V(P)$, i.e.,~$J-V(P) \in \T_{G-V(P)}(\kX)$. However it is given that~$\T_{G-V(P)}(\kX) = \maxset_w(\secl{G-V(P)}{k-1}{F \setminus V(P)})$, a contradiction. Hence~$J \in \maxset_w\{H \in \secl{G}{k}{F} \mid p \in N_G(H)\}$.
}
\end{proof}

\huir{Improve filltext? ``the problem to solve in \cref{alg:3way} reduces to the three problems solved in the recursive calls.''}The following lemma is used to argue that the branches of \cref{alg:3way} are disjoint.

\begin{lemma} \label{lem:3way-branching}
Let~$(G,k,F,T,w)$ be an almost leafless \ELSS{} instance such that~$G$ is connected and~$N_G(F) \neq \emptyset$. Fix some~$v \in N_G(T)$ and let~$P = (v=v_1,v_2,\ldots,v_\ell)$ be the unique maximal path disjoint from~$T$ satisfying~$\degree_G(v_i) = 2$ for each~$1 \leq i < \ell$ and~($v_\ell \in N_G(T)$ or~$\degree_G(v_\ell) > 2$). Then for any \emph{maximum-weight} $k$-secluded supertree~$H$ of~$F$, exactly one of the following holds:
\begin{enumerate}
    \item \label[condition]{lem:3way:v_ell}     ~$v_\ell \in N(H)$ (so~$v_\ell \notin F$),
    \item \label[condition]{lem:3way:path-split}~$|N(H) \cap V(P-F-v_\ell)| = 1$ and~$v_\ell \in V(H)$, or
    \item \label[condition]{lem:3way:path-incl} ~$V(P) \subseteq V(H)$.
\end{enumerate}
\end{lemma}
\begin{proof}
First note that such a vertex~$v$ exists since~$N_G(F) \neq \emptyset$ and~$G$ is connected, so~$N_G(T) \neq \emptyset$. Furthermore since the instance is almost leafless, \jjh{the path~$P$ is well defined.}\jjhr{I think we need almost leafless for well definedness of~$P$, the commented out sentence below is clear from the deifnition of $P$}
If there is no $k$-secluded supertree of~$F$, then there is nothing to show. So suppose~$H$ is a maximum-weight $k$-secluded supertree of~$F$. We have~$v \in V(P)$ is a neighbor of~$T \subseteq F \subseteq V(H)$, so either~$V(P) \subseteq V(H)$ or~$V(P)$ contains a vertex from~$N(H)$. In the first case \cref{lem:3way:path-incl} holds, in the second case we have~$|N(H) \cap V(P)| \geq 1$.
First suppose that~$|N(H) \cap V(P)| \geq 2$. Let~$i \in [\ell]$ be the smallest index such that~$v_i \in N(H) \cap V(P)$. Similarly let~$j \in [\ell]$ be the largest such index. \jjh{We show} that in this case we can contradict the fact that~$H$ is a maximum-weight $k$-secluded supertree of~$F$. \jjh{Observe that}~$H'= V(H) \cup \{v_i,\dots,v_{j-1}\}$ induces a tree since~$(v_i, \ldots, v_{j-1})$ forms a path of degree-2 vertices and the neighbor~$v_j$ of~$v_{j-1}$ is not in~$H$. \jjh{Furthermore~$H'$ has a strictly smaller neighborhood than $H$} and it has larger weight as vertices have positive weight. \jjh{Since~$F \subseteq V(H')$, this contradicts that~$H$ is a maximum-weight $k$-secluded supertree of~$F$.}\jjhr{reformulated sentence}

We conclude that~$|N(H) \cap V(P)| = 1$. Let~$i \in [\ell]$ be the unique index such that~$N(H) \cap V(P) = \{v_i\}$. Clearly~$v_i \notin F$. In the case that~$i = \ell$, then \cref{lem:3way:v_ell} holds. Otherwise if~$i < \ell$, the first condition of \cref{lem:3way:path-split} holds. In order to argue that the second condition also holds, suppose that~$v_\ell \notin V(H)$. Then~$H \cup \{v_i, \dots, v_{\ell-1}\}$ is a $k$-secluded supertree of~$F$ in~$G$ and it has larger weight than~$H$ as vertices have positive weight. This contradicts the fact that~$H$ has maximum weight, hence the second condition of \cref{lem:3way:path-split} holds as well.
\end{proof}

\hui{Armed with \cref{lem:combine,lem:combine-path,lem:3way-branching} we are now ready to prove correctness of the algorithm.}

\newcommand{\lemAlgcorrect}{
The algorithm described in \cref{sec:alg} is correct.
}
\begin{lemma}\label{lem:algcorrect}
\lemAlgcorrect
\end{lemma}
\begin{proof}
Let~$I = (G,k,F,T,w)$ be an \ELSS{} instance. We prove correctness by induction on~$|V(G) \setminus F|$. Assume the algorithm is correct for any input~$(\hat{G},\hat{k},\hat{F},\hat{T},\hat{w})$ with~$|V(\hat{G}) \setminus \hat{F}| < |V(G) \setminus F|$. 

\subsubsection*{Before \cref{alg:2nbrs}}
We first prove correctness when the algorithm terminates before \cref{alg:2nbrs}, which includes the base case of the induction. 
Note that if~$G[F]$ contains a cycle, then no induced subgraph~$H$ of~$G$ with~$F \subseteq V(H)$ can be acyclic. Therefore the set of maximum-weight $k$-secluded trees containing~$F$ is the empty set, so we correctly return~$\emptyset$.
Otherwise~$G[F]$ is acyclic. Clearly any connected component of~$G$ that has no vertices of~$F$ can be removed. If there are two connected components of~$G$ containing vertices of~$F$, then no induced subgraph of~$G$ containing all of~$F$ can be connected, again we correctly return the empty set. In the remainder we have that~$G$ is connected.

By iteratively applying \cref{lem:safe:contract} we conclude that a solution to the instance obtained after iteratively contracting (most) degree-1 vertices is also a solution to the original instance. Hence we can proceed to solve the new instance, which we know is almost leafless. In addition, observe that the contraction of degree-1 vertices maintains the property that~$G$ is connected and~$G[F]$ is acyclic.

After exhaustively adding vertices~$v \in N_G(T) \cap F$ to~$T$ we have that~$G[T]$ is a connected component of~$G[F]$. 
In the case that~$N_G(F) = \emptyset$, then since~$G$ is connected it follows that~$F = T = V(G)$ and therefore~$T$ is the only maximum-weight $k$-secluded tree. For any~$r \in V(G)$, the description~$(r, \emptyset)$ describes this $k$-secluded tree, so we return~$\{(r,\emptyset)\}$. In the remainder we have~$N_G(F) \neq \emptyset$.

Since~$N_G(F) \neq \emptyset$ and~$G$ is almost leafless, we argue that there is no $0$-secluded supertree of~$F$. Suppose~$G$ contains a $0$-secluded supertree~$H$ of~$F$, so~$|N_G(H)| = 0$ and since~$H \supseteq F$ is non-empty and~$G$ is connected we must have~$H=G$, hence~$G$ is a tree with at least two vertices (since~$F$ and~$N_G(F)$ are both non-empty) so~$G$ contains at least two vertices of degree-1, contradicting that~$G$ is almost leafless. So there is no $k$-secluded supertree of~$F$ in~$G$ and the algorithm correctly returns~$\emptyset$ if~$k = 0$.

Observe that the value~$|V(G) \setminus F|$ cannot have increased since the start of the algorithm since we never add vertices to~$G$ and any time we remove a vertex from~$F$ it is also removed from~$G$. Hence we can still assume in the remainder of the proof that the algorithm is correct for any input~$(\hat{G},\hat{k},\hat{F},\hat{T},\hat{w})$ with~$|V(\hat{G}) \setminus \hat{F}| < |V(G) \setminus F|$.
To conclude this part of the proof, we have established that if the algorithm terminates before reaching \cref{alg:2nbrs}, then its output is correct. On the other hand, if the algorithm continues we can make use of the following properties of the instance just before reaching \cref{alg:2nbrs}:
\begin{property} \label{prop:before-step1}
If the algorithm does not terminate before reaching \cref{alg:2nbrs} then
 \begin{enumerate*}[label=(\roman*)]
  \item the \ELSS{} instance~$(G,k,F,T,w)$ is almost leafless,
  \item $G[F]$ is acyclic,
  \item $G[T]$ is a connected component of~$G[F]$,
  \item $G$ is connected,
  \item $k > 0$, and
  \item $N_G(F) \neq \emptyset$.
 \end{enumerate*}
\end{property}

\subsubsection*{\Cref{alg:2nbrs}}
Before arguing that the return value in \cref{alg:2nbrs} is correct, we observe the following.

\begin{claim} \label{claim:2nbrs}
If~$H$ is an induced subtree of~$G$ that contains~$T$ and~$v \in N_G(T)$ has at least two neighbors in~$T$, then~$v \in N_G(H)$.
\end{claim}
\begin{claimproof}
    Suppose~$v \not\in N_G(H)$, then since~$v \in N_G(T)$ and~$T \subseteq V(H)$ we have that~$v \in V(H)$. But then since~$T$ is connected, subgraph~$H$ contains a cycle. This contradicts that~$H$ is a tree and confirms that~$v \in N_G(H)$.
\end{claimproof}

Now consider the case that in \cref{alg:2nbrs} we find a vertex~$v \in N_G(T)$ with two neighbors in~$T$, and let~$\kX'$ be the set of descriptions as obtained by the algorithm through recursively solving the instance~$(G-v,k-1,F,T,w)$. Since~$|V(G-v)\setminus F| < |V(G) \setminus F|$ (as~$v \not\in F$) we know by induction that~$\T_{G-v}(\kX')$ is the set of all maximum-weight $(k-1)$-secluded supertrees of~$F$ \bmp{in~$G-v$}. 
\bmp{Any~$H \in \T_{G-v}(\kX')$ is} an induced subtree of~$G$ with~$T \subseteq V(H)$, so by \cref{claim:2nbrs} we have~$v \in N_G(H)$ for all~$H \in \T_{G-v}(\kX')$. We can now apply \cref{lem:combine} to conclude that~$\T_G(\{(r,\X \cup \{\{v\}\}) \mid (r,\X) \in \kX'\})$ is the set of all maximum-weight $k$-secluded supertrees~$H$ of~$F$ in~$G$ for which~$v \in N_G(H)$. Again by \cref{claim:2nbrs} we have that~$v \in N_G(H)$ for all such $k$-secluded supertrees of~$F$, hence~$\T_G(\{(r,\X \cup \{\{v\}\}) \mid (r,\X) \in \kX'\})$ is the set of all maximum-weight $k$-secluded supertrees of~$F$ in~$G$. We argue non-redundancy of the output. Suppose that two descriptions~$(r,\X \cup \{\{v\}\})$ and~$(r',\X' \cup \{\{v\}\})$ describe the same supertree~$H$ of~$F$ in~$G$. Note that then~$(r,\X)$ and~$(r',\X)$ describe the same supertree~$H$ of~$F$ in~$G-v$, which contradicts the induction hypothesis that the output of the recursive call was correct and therefore non-redundant.



Concluding this part of the proof, we showed that if the algorithm terminates during \cref{alg:2nbrs}, then its output is correct. On the other hand, if the algorithm continues after \cref{alg:2nbrs} we can make use of the following in addition to \cref{prop:before-step1}.
\begin{property} \label{prop:step1}
If the algorithm does not terminate before reaching \cref{alg:greedy} then no vertex~$v \in N_G(T)$ has two neighbors in~$T$.
\end{property}

\subsubsection{\cref{alg:greedy}}
In \cref{alg:greedy} we use \cref{lem:greedy:naive} if~$|N_G(T)| > k(k+1)$. The preconditions of the lemma are satisfied \bmp{since}~$k > 0$ and the instance is almost leafless by \cref{prop:before-step1}. If it concludes that~$G$ does not contain a $k$-secluded supertree of~$F$, then the algorithm correctly outputs~$\emptyset$. Otherwise it finds a vertex~$v \in V(G)\setminus F$ such that any $k$-secluded supertree~$H$ of~$F$ in~$G$ satisfies~$v \in N_G(H)$. 
We argue that the algorithm's output is correct. Let~$\kX'$ be the set of descriptions as obtained \bmp{through} recursively solving~$(G-v,k-1,F,T,w)$. Since~$v \not\in F$ we have~$|(V(G-v) \setminus F| < |V(G) \setminus F|$, so by induction we have that~$\T_{G-v}(\kX')$ is the set of all maximum-weight $(k-1)$-secluded supertrees of~$F$ in~$G-v$. 
Furthermore by \cref{obs:secluded} for any~$H \in \T_{G-v}(\kX') = \secl{G-v}{k-1}{F}$ we have~$H \in \secl{G}{k}{F}$, and therefore~$v \in N_G(H)$. It follows that \cref{lem:combine} applies to~$\kX'$ so we can conclude that~$\T_G(\{(r, \X \cup \{\{v\}\}) \mid (r,\X) \in \kX'\})$ is the set of maximum-weight $k$-secluded supertrees~$H$ of~$F$ in~$G$ for which~$v \in N_G(H)$. Since we know there are no $k$-secluded supertrees~$H$ of~$F$ in~$G$ for which~$v \not\in N_G(H)$, it follows that~$\T_G(\{(r, \X \cup \{\{v\}\}) \mid (r,\X) \in \kX\})$ is the set of maximum-weight $k$-secluded supertrees of~$F$ in~$G$ as required. 
Non-redundancy of the output follows as in \cref{alg:2nbrs}.



To summarize the progress so far, we have shown that if the algorithm terminates before it reaches \cref{alg:3way}, then its output is correct. Alternatively, if we proceed to \cref{alg:3way} we can make use of the following property, in addition to \cref{prop:before-step1,prop:step1}\bmp{, which we will use later in the running time analysis}.
\begin{property} \label{prop:step3}
 If the algorithm does not terminate before reaching \cref{alg:3way} then~$|N_G(T)| \leq k(k+1)$.
\end{property}

\subsubsection{\Cref{alg:3way}}

Fix some~$v \in N_G(T)$, which exists as~$N_G(T) \neq \emptyset$ by \cref{prop:before-step1}. Let~$P = (v = v_1, \dots, v_\ell)$ be a path as described in \cref{lem:3way-branching}. By~\cref{lem:3way-branching} we can partition the set~$\maxset_w (\secl{G}{k}{F})$ of maximum-weight $k$-secluded supertrees of~$F$ in~$G$ into the following three sets:
\begin{itemize}
    \item $\mathcal{T}_1 = \{H \in \maxset_w (\secl{G}{k}{F}) \mid v_\ell \in N_G(H)\}$,
    \item $\mathcal{T}_2 = \{H \in \maxset_w (\secl{G}{k}{F}) \mid |N(H) \cap V(P - F - v_\ell)| = 1 \text{ and } v_\ell \in V(H)\}$, 
    \item $\mathcal{T}_3 = \{H \in \maxset_w (\secl{G}{k}{F}) \mid V(P) \subseteq V(H)\}$, 
\end{itemize}

Consider the sets~$\kX_1$,~$\kX_2$, and~$\kX_3$ of descriptions as obtained through recursion in \cref{alg:3way} of the algorithm. By induction we have the following:
\begin{itemize}
    \item $\T_{G-v_\ell}(\kX_1) = \maxset_w (\secl{G-v_\ell}{k-1}{F})$, since~$|V(G-v_\ell) \setminus F| < |V(G) \setminus F|$,
    \item $\T_{G-V(P - v_\ell)}(\kX_2) = \maxset_w (\secl{G-V(P - v_\ell)}{k-1}{(F \setminus V(P)) \cup \{v_\ell\}})$ since
    \begin{align*} & |V(G-V(P - v_\ell)) \setminus ((F \setminus V(P)) \cup \{v_\ell\})| = \\ & |V(G-V(P - v_\ell)) \setminus (F \cup \{v_\ell\})| <  |V(G) \setminus F|, \mbox{ and} \end{align*}
    \item $\T_{G}(\kX_3) = \maxset_w (\secl{G}{k}{F \cup V(P)})$ since~$|V(G) \setminus (F \cup V(P))| < |V(G) \setminus F|$.
\end{itemize}

Let~$\kX'_1$,~$\kX'_2$, and~$\kX'_3$ be the sets of descriptions as \bmp{computed} in \cref{alg:3way} of the algorithm. 

\begin{claim}\label{claim:branchdescriptions}
The sets~$\kX'_1$,~$\kX'_2$, and~$\kX'_3$ consist of valid descriptions for~$G$.
\end{claim}
\begin{claimproof}
\hui{To argue that~$\kX'_1 = \{(r,\X \cup \{\{v_\ell\}\}) \mid (r,\X) \in \kX_1\}$ consists of valid descriptions for~$G$ we show for an arbitrary description~$(r,\X) \in \kX_1$ for~$G-v_\ell$ that~$(r,\X \cup \{\{v_\ell\}\})$ is a valid description for~$G$.}
Clearly~$r \in V(G)$ and~$\X \cup \{\{v_\ell\}\}$ consists of pairwise disjoint subsets of~$V(G-r)$. Consider any set~$S'$ consisting of exactly one vertex from each set~$X \in \X \cup \{\{v_\ell\}\}$. Clearly~$v_\ell \in S'$. Let~$S = S' \setminus \{v_\ell\}$. Since~$(r, \X)$ is a description \hui{for~$G-v_\ell$}, the connected component~$H$ of~$G-v_\ell - S$ containing~$r$ is acyclic and satisfies~$N_{G-v_\ell}(H) = S$. Note that the connected component of~$G-S'$ containing~$r$ is identical to~$H$ and is therefore also acyclic. All that is left to argue is that~$v_\ell \in N_G(H)$. If~$\ell=1$, then~$v_\ell \in N_G(F)$ and the claim follows as~$H$ is a supertree of~$F$.
Otherwise \hui{note that since~$(r,\X) \in \kX_1$ we have that~$H \in \maxset_w(\secl{G-v_\ell}{k-1}{F})$, i.e.,~$H$ is of maximum weight. \bmp{Since all vertices of~$V(P - v_\ell)$ have degree~$2$ in~$G$, with~$v_{\ell-1}$ adjacent to~$v_\ell$, the graph}~$G[V(H) \cup V(P-v_\ell)]$ is acyclic and~$|N_G(H)| = |N_G(V(H) \cup V(P-v_\ell))|$. It follows that~$V(P-v_\ell) \subseteq V(H)$ since otherwise the secluded tree~$G[V(H) \cup V(P - v_\ell)]$ would have larger weight than~$H$. Hence~$v_{\ell-1} \in V(H)$ so~$v_\ell \in N_G(H)$.} 

Next \hui{we argue~$\kX'_2$ consists of valid descriptions for~$G$. Recall that~$\kX'_2 = \{(r, \X \cup \{M\}) \mid (r,\X) \in \kX_2\}$ where~$M$ is the set of minimum weight vertices in~$P-F-v_\ell$, so it suffices to show for an arbitrary description~$(r, \X) \in \kX_2$ for~$G-V(P-v_\ell)$ that~$(r, \X \cup \{M\}) \in \kX'_2$ is a valid description for~$G$.} Again it is easy to see that~$r \in V(G)$ and~$\X \cup \{M\}$ consists of pairwise disjoint subsets of~$V(G-r)$. Consider any set~$S'$ consisting of exactly one vertex from each set~$X \in \X \cup \{M\}$. Let~$S = S' \setminus M$ and~$\{m\} = M \cap S'$. Since~$(r, \X)$ is a description \hui{for~$G-V(P-v_\ell)$}, the connected component~$H$ of~$G-V(P - v_\ell) - S$ containing~$r$ is acyclic and satisfies~$N_{G-V(P - v_\ell)}(H) = S$. Note that the connected component~$H'$ of~$G-S'$ containing~$r$ is a supergraph of~$H$ and so~$S \subseteq N_G(H')$. All that is left to argue is that~$H'$ is acyclic and~$m \in N_G(H')$. Let~$u$ be the vertex in~$T$ that is adjacent to~$v_1$. Note that this vertex is uniquely defined since no vertex in~$N_G(T)$ has two neighbors in~$T$. Since~$H$ is a supertree of~$(F \setminus V(P)) \cup \{v_\ell\}$ and~$u, v_\ell \in F$, it follows that~$P - v_\ell$ is a path between two vertices in~$H$\bmp{, of which~$S'$ contains exactly one vertex chosen from~$M$. Consequently, the component~$H'$ of~$G-S$ } satisfies~$V(H') = V(H) \cup V(P - v_\ell - m)$,~$H'$ is acyclic, and~$m \in N_G(H')$. 

Finally since~$\kX_3$ is a set of descriptions for~$G$ and~$\kX'_3 = \kX_3$, the claim holds for~$\kX'_3$.
\end{claimproof}

Before we proceed to show that the output of the algorithm is correct, we prove two claims about intermediate results obtained by modifying the output of a recursive call.

\begin{claim} \label{claim:3way:X1}
 $\T_{G}(\kX'_1) = \maxset_w\{H \in \secl{G}{k}{F} \mid v_\ell \in N_G(H)\}$
\end{claim}
\begin{claimproof}
 Recall~$\kX'_1$ is defined as~$\{(r,\mathcal{X} \cup \{\{v_\ell\}\}) \mid (r,\mathcal{X}) \in \mathfrak{X}_1\}$. 
 We know that~$\T_{G-v_\ell}(\kX_1) = \maxset_w(\secl{G-v_\ell}{k-1}{F})$. In order to apply \cref{lem:combine} we prove that~$v_\ell \in N_G(H)$ for all~$H \in \T_{G-v_\ell}(\kX_1)$. Let~$H \in \T_{G-v_\ell}(\kX_1)$ be arbitrary. If~$\ell = 1$, then as~$v_\ell = v \in N_G(T)$ and~$T \subseteq H$ we get~$v_\ell \in N_G(H)$. Otherwise if~$\ell > 1$, suppose for the sake of contradiction that~$v_{\ell-1} \notin V(H)$. Then some vertex~$u \in V(P-F-v_\ell)$ must be contained in~$N_{G-v_\ell}(H)$. But then observe that~$H \cup V(P-v_\ell)$ acyclic and has strictly larger weight than~$H$, while~$|N_{G-v_\ell}(H \cup V(P-v_\ell))| < |N_{G-v_\ell}(H)|$. This contradicts the choice of~$H$. It follows that~$v_{\ell-1} \in V(H)$ and therefore~$v_\ell \in N_G(H)$.
 We can now apply \cref{lem:combine} to obtain that~$\T_{G}(\kX'_1) = \maxset_w\{H \in \secl{G}{k}{F} \mid v_\ell \in N_G(H)\}$. 
\end{claimproof}

\begin{claim} \label{claim:3way:X2}
 If~$m \in V(P-F-v_\ell)$ then~$\T_G(\{(r,\X \cup \{\{m\}\}) \mid (r,\X) \in \kX_2\}) = \maxset_w\{\hat{H} \in \secl{G}{k}{F \cup \{v_\ell\}} \mid m \in N_G(\hat{H})\}$.
\end{claim}
\begin{claimproof}
 We show \cref{lem:combine-path} applies to the instance~$(G,k,F \cup \{v_\ell\}, T, w)$.
 Recall that by induction~$\T_{G-V(P-v_\ell)}(\kX_2) = \maxset_w(\secl{G-V(P-v_\ell)}{k-1}{(F \setminus V(P)) \cup \{v_\ell\}}) =  \maxset_w(\secl{G-V(P-v_\ell)}{k-1}{(F \cup \{v_\ell\}) \setminus V(P - v_\ell)})$.
 Recall also that~$v_1 = v \in N_G(T)$ has exactly one neighbor in~$T \subseteq F$ by \cref{prop:step1}, let~$v'$ be this vertex. Observe that~$P-v_\ell$ is a path in~$G$ with~$\degree_G(p) = 2$ for all~$p \in V(P-v_\ell)$ and~$N_G(P-v_\ell) = \{v', v_\ell\}$ with~$v',v_\ell \in F \cup \{v_\ell\}$.
 Then since~$m \in V(P-F-v_\ell) = V(P-v_\ell) \setminus F$ we can apply \cref{lem:combine-path} to obtain~$\T_G(\{(r,\X \cup \{\{m\}\}) \mid (r,\X) \in \kX_2\}) = \maxset_w\{\hat{H} \in \secl{G}{k}{F \cup \{v_\ell\}} \mid m \in N_G(\hat{H})\}$.
\end{claimproof}

We now show that all maximum-weight $k$-secluded supertrees of~$F$ in~$G$ are described by some description in our output. More formally, we show that~$\maxset_w(\secl{G}{k}{F}) \subseteq \T_G(\kX'_1 \cup \kX'_2 \cup \kX'_3) \subseteq \secl{G}{k}{F}$. 
To that end we first show that~$\T_i \subseteq \T_G(\kX'_i) \subseteq \secl{G}{k}{F}$ for each~$i \in [3]$ in \cref{claim:3way:T1,claim:3way:T2,claim:3way:T3}.

\begin{claim} \label{claim:3way:T1}
 $\T_1 \subseteq \T_G(\kX'_1) \subseteq \secl{G}{k}{F}$
\end{claim}
\begin{claimproof}
 It follows from \cref{claim:3way:X1} that~$\T_G(\kX'_1) \subseteq \secl{G}{k}{F}$ so it remains to show that~$\T_1 \subseteq \T_{G}(\kX'_1)$.
 Let~$H \in \T_1$ be arbitrary. By definition of~$\T_1$ we have~$H \in \maxset_w(\secl{G}{k}{F}) \subseteq \secl{G}{k}{F}$ and~$v_\ell \in N_G(H)$. So then clearly~$H \in \maxset_w\{H' \in \secl{G}{k}{F} \mid v_\ell \in N_G(H')\} = \T_G(\kX'_1)$ (by \cref{claim:3way:X1}). Since~$H \in \T_1$ was arbitrary we conclude~$\T_1 \subseteq \T_G(\kX'_1)$ completing the proof.
\end{claimproof}

\begin{claim} \label{claim:3way:T2}
 $\T_2 \subseteq \T_G(\kX'_2) \subseteq \secl{G}{k}{F}$
\end{claim}
\begin{claimproof}
 Recall~$\kX'_2$ is defined as~$\{(r,\mathcal{X} \cup \{M\}) \mid (r,\mathcal{X}) \in \mathfrak{X}_2\}$, where~$M$ is the set of minimum weight vertices in~$P - F - v_\ell$.
 For any~$u \in V(P-F-v_\ell)$ we define~$\kX_2^u := \{(r, \X \cup \{\{u\}\}) \mid (r,\X) \in \kX_2\}$.
 By repeated application of \cref{obs:descr-merge} we have that~$\T_G(\bigcup_{u \in M} \kX_2^u) = \T_G(\{(r,\X \cup \{M\}) \mid (r,\X) \in \kX_2\})$. 
 \jjh{Observe that~$\kX_2^u$ is a valid set of descriptions for~$G$ for each~$u \in M$.}
 By \cref{obs:descr-union} and definition of~$\kX'_2$ we have~$\bigcup_{u \in M} \T_G(\kX_2^u) = \T_G(\bigcup_{u \in M} \kX_2^u) = \T_G(\kX'_2)$.
 We will prove~$\T_2 \subseteq \bigcup_{u \in M} \T_G(\kX_2^u) \subseteq \secl{G}{k}{F}$.
 
 We show~$\T_2 \subseteq \bigcup_{u \in M} \T_G(\kX_2^u)$. Let~$H \in \T_2$ be arbitrary. By definition of~$\T_2$ we have~$H \in \maxset_w(\secl{G}{k}{F})$,~$|N_G(H) \cap V(P-F-v_\ell)| = 1$, and~$v_\ell \in V(H)$. Then also~$H \in \maxset_w(\secl{G}{k}{F \cup \{v_\ell\}})$.
 Let~$m$ be such that~$N_G(H) \cap V(P-F-v_\ell) = \{m\}$.
 Since~$m \in V(P-F-v_\ell)$, by \cref{claim:3way:X2} we have that~$\T_G(\kX^m_2) = \maxset_w\{\hat{H} \in \secl{G}{k}{F \cup \{v_\ell\}} \mid m \in N_G(\hat{H})\}$.
 Since~$m \in N_G(H)$ and~$H \in \maxset_w(\secl{G}{k}{F \cup \{v_\ell\}})$ clearly~$H \in \maxset_w\{\hat{H} \in \secl{G}{k}{F \cup \{v_\ell\}} \mid m \in N_G(\hat{H})\} = \T_G(\kX_2^m)$.
 It follows that~$H \in \bigcup_{u \in M} \T_G(\kX_2^u)$. Since~$H$ was arbitrary we conclude that~$\T_2 \subseteq \bigcup_{u \in M} \T_G(\kX_2^u)$.
 
 It remains to show that~$\bigcup_{u \in M} \T_G(\kX_2^u) \subseteq \secl{G}{k}{F}$.
 Let~$u \in M$ be arbitrary. We show~$\T_{G}(\kX_2^u) \subseteq \secl{G}{k}{F}$. It suffices to show that when considering a set~$S$ consisting of one element from each set of a description~$(r, \mathcal{X} \cup \{u\}) \in \kX_2^u$, the component~$H$ of~$G-S$ containing~$r$ is a $k$-secluded supertree of~$F$ in~$G$. This component~$H$ is a $k$-secluded tree in~$G$ since~$\kX_2^u$ is a valid set of descriptions for~$G$ of order at most~$k$. It remains to show that~$F \subseteq V(H)$. 
 By induction the component of~$(G - V(P - v_\ell)) - (S \setminus \{u\})$ contains all of~$(F \setminus V(P)) \cup \{v_\ell\}$.
 As the two neighbors of~$V(P - v_\ell)$ both belong to~$F \cup \{v_\ell\}$, the subpath of~$P$ before~$u$ and subpath after~$u$ are both reachable from~$r$ in~$G-S$.
 Hence~$V(P-u) \subseteq V(H)$, and since~$u \in M \subseteq V(P-F-v_\ell)$ we know~$u \not\in F$, so~$F \subseteq V(H)$.
 It follows that~$H$ is a $k$-secluded supertree of~$F$ in~$G$ so since~$H \in \T_{G}(\kX_2^u)$ was arbitrary we have~$\T_{G}(\kX_2^u) \subseteq \secl{G}{k}{F}$.
\end{claimproof}

\begin{claim} \label{claim:3way:T3}
 $\T_3 \subseteq \T_G(\kX'_3) \subseteq \secl{G}{k}{F}$
\end{claim}
\begin{claimproof}
 Recall~$\kX'_3$ is defined to be equal to~$\kX_3$, so we show~$\T_3 \subseteq \T_G(\kX_3) \subseteq \secl{G}{k}{F}$.
 
 Let~$H \in \T_3$ be arbitrary. By definition of~$\T_3$ we have~$H \in \maxset_w(\secl{G}{k}{F}) \subseteq \secl{G}{k}{F}$ and~$V(P) \subseteq V(H)$. So clearly~$H \in \secl{G}{k}{F \cup V(P)}$. To show that~$H \in \T_G(\kX_3) = \maxset_w(\secl{G}{k}{F \cup V(P)})$ we have to show for all~$H' \in \secl{G}{k}{F \cup V(P)}$ that~$w(H') \leq w(H)$. Suppose for contradiction there is an~$H' \in \secl{G}{k}{F \cup V(P)}$ such that~$w(H') > w(H)$. Clearly~$H' \in \secl{G}{k}{F}$ but then~$w(H') > w(H)$ contradicts that~$H \in \maxset_w(\secl{G}{k}{F}))$. So by contradiction it follows that~$H \in \maxset_w(\secl{G}{k}{F \cup V(P)}) = \T_G(\kX_3)$ and since~$H \in \T_3$ was arbitrary we conclude~$\T_3 \subseteq \T_G(\kX_3)$.
 
 Finally observe that~$\T_G(\kX_3) = \maxset_w(\secl{G}{k}{F \cup V(P)}) \subseteq \secl{G}{k}{F \cup V(P)} \subseteq \secl{G}{k}{F}$, completing the proof.
\end{claimproof}

It follows from \cref{claim:3way:T1,claim:3way:T2,claim:3way:T3} that~$\maxset_w(\secl{G}{k}{F}) = \T_1 \cup \T_2 \cup \T_3 \subseteq \T_G(\kX'_1) \cup \T_G(\kX'_2) \cup \T_G(\kX'_3) \subseteq \secl{G}{k}{F}$. So then we have
\begin{equation} \label{eqn:T123}
 \maxset_w(\secl{G}{k}{F}) = \maxset_w(\T_G(\kX'_1) \cup \T_G(\kX'_2) \cup \T_G(\kX'_3)) .
\end{equation}

The algorithm proceeds to calculate values~$w_1, w_2, w_3$ based on an arbitrary secluded tree in~$\T_G(\kX'_1)$,~$\T_G(\kX'_2)$, and~$\T_G(\kX'_3)$ respectively. We show that, for any~$i \in [3]$, all secluded trees in~$\T_G(\kX'_i)$ have weight~$w_i$.
\begin{itemize}
 \item For~$i=1$, we know from the proof of \cref{claim:3way:X1} that~$\T_{G}(\kX'_1) = \maxset_w\{H \in \secl{G}{k}{F} \mid v_\ell \in N_G(H)\}$, and clearly all trees in~$\maxset_w\{H \in \secl{G}{k}{F} \mid v_\ell \in N_G(H)\}$ have the same weight, which must be~$w_1$.
 
 \item For~$i=2$, consider two arbitrary secluded trees~$H_1,H_2 \in \T_G(\kX'_2) = \T_G(\{(r, \X \cup \{M\}) \mid (r, \X) \in \kX_2\})$ where~$M \subseteq V(P-F-v_\ell)$ is the set of minimum weight vertices in~$P-F-v_\ell$. We show that~$w(H_1) = w(H_2)$. Observe that~$N_G(H_1) \cap M = \{m_1\}$ for some~$m_1 \in V(P-v_\ell)\setminus F$, so~$H_1 \in \T_G(\{(r,\X \cup \{\{m_1\}\}) \mid (r,\X) \in \kX_2\})$. Hence~$H_1 \in \maxset_w\{\hat{H} \in \secl{G}{k}{F \cup \{v_\ell\}} \mid m_1 \in N_G(\hat{H})\}$ by \cref{claim:3way:X2} since~$m_1 \in V(P-v_\ell) \setminus F$. Similarly~$H_2 \in \maxset_w\{\hat{H} \in \secl{G}{k}{F \cup \{v_\ell\}} \mid m_2 \in N_G(\hat{H})\}$ for some~$m_2 \in V(P - v_\ell) \setminus F$. 
 Consider the graph~$H'_2 := G[V(H_2) \cup \{m_2\}] - m_1$ and observe that~$H'_2 \in \{\hat{H} \in \secl{G}{k}{F \cup \{v_\ell\}} \mid m_1 \in N_G(\hat{H})\}$. Additionally~$w(H_2) = w(H'_2) + w(m_2) - w(m_1) = w(H'_2)$ since~$m_1,m_2$ are of minimum weight in~$V(P-F-v_\ell)$, so~$H'_2 \in \maxset_w\{\hat{H} \in \secl{G}{k}{F \cup \{v_\ell\}} \mid m_1 \in N_G(\hat{H})\}$. It follows that~$w(H_2) = w(H'_2) = w(H_1)$.
 Since~$H_1,H_2 \in \T_G(\kX'_2)$ are arbitrary, we have that all secluded trees in~$\T_G(\kX'_2)$ have the same weight, which must be~$w_2$. \bmpr{I can follow the argument and agree that it does the job. I don't find it the most insightful, though, and wonder whether you see any advantage in the following alternative. Rather than taking two distinct trees being described and comparing their weight by transforming one in the other, an alternative would be to say that for any tree~$H_1$ being described in~$G - V(P - v_\ell)$ as the component of~$G - V(P - v_\ell) - S$ containing~$r$, the component of~$(G - S) - \{m\}$ containing~$r$ is exactly~$w(P) - w(m)$ heavier. Since~$m$ has minimum weight~$w_0$ and this goes for all trees being described, all trees being described in~$G$ are exactly~$w(P) - w_0$ heavier than the tree in~$G- V(P - v_\ell)$ from which their description was generated, and hence all trees being described still have the same weight. This argument could be backed up by the following statement, which we could add as an observation somewhere and invoke in multiple places (like in the proof of Claim 4): Let~$G$ be a graph and~$P$ be a path of degree-2 vertices in~$G$ with~$N_G (P)=\{a,b\}$. If~$S\subseteq V(G) \setminus V(P)$ and~$r \notin V(P) \cup S$ such that the component~$H$ of~$(G-V(P))-S$ containing~$r$ is acyclic and contains~$\{a,b\}$, then for any~$m \in V(P)$ the component of~$(G-S)-m$ containing~$r$ is acyclic, has vertex set~$H \cup (V(P) \setminus \{m\})$, and is therefore adjacent to~$m$.}
 
 \item For~$i=3$ we know that~$\T_G(\kX'_3) = \T_G(\kX_3) = \maxset_w(\secl{G}{k}{F \cup V(P)})$, and clearly all trees in~$\maxset_w(\secl{G}{k}{F \cup V(P)})$ have the same weight, which must be~$w_3$.
\end{itemize}

Clearly it follows that
\begin{align*}
 \T_G(\kX')
 &= \T_G\left (\bigcup_{\{i \in [3] \mid w_i=\max\{w_1,w_2,w_3\}\}} \kX'_i \right)
	&&\text{Definition of~$\kX'$ in \cref{alg:3way}} \\
 &= \maxset_w(\T_G(\kX'_1 \cup \kX'_2 \cup \kX'_3))
	&&\text{Any tree in~$\T_G(\kX'_i)$ has weight~$w_i$.}\\
 &= \maxset_w(\T_G(\kX'_1) \cup \T_G(\kX'_2) \cup \T_G(\kX'_3))
	&&\text{\cref{obs:descr-union}}\\
 &= \maxset_w(\secl{G}{k}{F}) .
	&&\text{\cref{eqn:T123}}
\end{align*}

Hence the algorithm correctly returns a set of descriptions~$\kX'$ for which~$\T_G(\kX') = \maxset_w(\secl{G}{k}{F})$. To complete the proof of correctness, we show that~$\kX'$ is non-redundant for~$G$.

\begin{claim}
~$\kX'$ is non-redundant for~$G$.
\end{claim}
\begin{claimproof}
 Suppose for contradiction that~$\kX'$ is redundant for~$G$, i.e., there is a $k$-secluded tree~$H$ in~$G$ such that~$H$ is described by two distinct descriptions~$(r_1,\X_1), \linebreak[1] (r_2,\X_2) \in \kX'$. Since~$\kX' \subseteq \kX'_1 \cup \kX'_2 \cup \kX'_3$ it suffices to consider the cases~$(r_1, \X_1) \in \kX'_1$,~$(r_1, \X_1) \in \kX'_2$, and~$(r_1, \X_1) \in \kX'_3$.
 
 \begin{itemize}
  \item If~$(r_1, \X_1) \in \kX'_1$, then~$\{v_\ell\} \in \X_1$ by definition of~$\kX'_1$. So then~$v_\ell \in N_G(H)$. Since~$H$ is described by~$(r_2, \X_2)$ there exists~$X \in \X_2$ such that~$v_\ell \in X$.
  If~$(r_2, \X_2) \in \kX'_3 = \kX_3$ then~$X$ must be part of a description in~$\kX_3$, so there must be a~$H' \in \T_{G}(\kX_3)$ with~$v_\ell \in N_G(H')$. However we know by induction that~$\T_{G}(\kX_3) = \maxset_w(\secl{G}{k}{F \cup V(P)})$, so~$v_\ell \in F \cup V(P) \subseteq V(H'')$ for all~$H'' \in \T_{G-V(P-v_\ell)}(\kX_3)$. It follows that~$(r_2, \X_2) \not\in \kX'_3$.
  If~$(r_2, \X_2) \in \kX'_2$ then we show this also leads to a contradiction. It follows from the definition of~$\kX'_2$ that~$X$ must be part of a description in~$\kX_2$, however we know by induction that~$\T_{G-V(P-v_\ell)}(\kX_2) = \maxset_w(\secl{G-V(P-v_\ell)}{k-1}{(F \setminus V(P)) \cup \{v_\ell\}})$, so~$v_\ell \in (F \setminus V(P)) \cup \{v_\ell\} \subseteq V(H'')$ for all~$H'' \in \T_{G-V(P-v_\ell)}(\kX_2)$. Hence~$(r_2, \X_2) \not\in \kX'_2$, and since also~$(r_2, \X_2) \not\in \kX'_3$ we have that~$(r_2, \X_2) \in \kX'_1$, meaning~$X = \{v_\ell\}$.
  
  Since~$(r_1,\X_1)$ and~$(r_2,\X_2)$ are distinct,~$\{v_\ell\} \in \X_1$, and~$\{v_\ell\} \in \X_2$ we have that~$(r_1, \X_1 \setminus \{\{v_\ell\}\})$ and~$(r_2, \X_2 \setminus \{\{v_\ell\}\})$ are distinct. Observe that~$(r_1, \X_1 \setminus \{\{v_\ell\}\}) \in \kX_1$ and~$(r_2, \X_2 \setminus \{\{v_\ell\}\}) \in \kX_1$. We know by induction that~$\kX_1$ is non-redundant for~$G-v$. However since~$H$ is a $k$-secluded tree in~$G$ with~$v_\ell \in N_G(H)$ we have that~$H$ is a $(k-1)$-secluded tree in~$G-v_\ell$, and clearly~$H$ is described by both~$(r_1, \X_1 \setminus \{\{v_\ell\}\})$ and~$(r_2, \X_2 \setminus \{\{v_\ell\}\})$, contradicting that~$\kX_1$ is non-redundant for~$G-v$.
  
  \item If~$(r_1, \X_1) \in \kX'_2$, then without loss of generality we can assume that~$(r_2, \X_2) \not\in \kX'_1$ since otherwise we can swap the roles of~$(r_1, \X_1)$ and~$(r_2, \X_2)$ and the previous case would apply. Suppose that~$(r_2, \X_2) \in \kX'_3 = \kX_3$, then~$H \in \T_G(\kX_3)$ and we have by induction that~$\T_G(\kX_3) = \maxset_w(\secl{G}{k}{F \cup V(P)})$ hence~$v_\ell \in F \cup V(P) \subseteq V(H)$. This contradicts~$v_\ell \in N_G(H)$ so~$(r_2, \X_2) \not\in \kX'_3$. This leaves as only option that~$(r_2, \X_2) \in \kX'_2$.
  
  Recall that~$\kX'_2 = \{(r, \X \cup \{M\}) \mid (r,\X) \in \kX_2\}$ where~$M \subseteq V(P-F-v_\ell)$, so~$(r_1, \X_1 \setminus \{M\}) \in \kX_2$ and~$(r_2, \X_2 \setminus \{M\}) \in \kX_2$. Since~$\kX_2$ is a set of valid descriptions for~$G-V(P-v_\ell)$ (by induction) we have that~$\X_1 \setminus \{M\}$ and~$\X_2 \setminus \{M\}$ contain only subsets of~$V(G) \setminus V(P-v_\ell)$, so~$N_G(H)\setminus M \subseteq V(G) \setminus V(P-v_\ell)$. Observe that since the path~$P-v_\ell$ is connected to~$H' := H - V(P-v_\ell)$ only via its endpoints, and~$H$ does not contain~$m \in V(P)$ we have that~$H'$ remains connected so~$H'$ is a $(k-1)$-secluded tree in~$G-V(P-v_\ell)$ described by~$(r_1,\X_1 \setminus \{M\})$ as well as~$(r_2,\X_2 \setminus \{M\})$. However this contradicts that~$\kX_2$ is a non-redundant set of descriptions for~$G-V(P-v_\ell)$ as given by induction.
  
  \item If~$(r_1, \X_1) \in \kX'_3 = \kX_3$, then without loss of generality we can assume~$(r_2, \X_2) \in \kX'_3 = \kX_3$ since otherwise we can swap the roles of~$(r_1, \X_1)$ and~$(r_2, \X_2)$ and one of the previous cases would apply. But then~$H$ is a $k$-secluded tree in~$G$ described by two distinct descriptions from~$\kX_3$, i.e.~$\kX_3$ is redundant for~$G$ contradicting the induction hypothesis. \qedhere
 \end{itemize}
\end{claimproof}

This concludes the proof of \cref{lem:algcorrect} and establishes correctness.
\end{proof}

\newcommand{\strace}{trace\xspace}
\newcommand{\straces}{traces\xspace}
\subsection{Runtime analysis} \label{sec:time}

If all recursive calls in the algorithm would decrease~$k$ then, since for~$k=0$ it does not make any further recursive calls, the maximum recursion depth is~$k$. However in \cref{alg:3way}(c) the recursive call does not decrease~$k$. In order to bound the recursion depth, we show the algorithm cannot make more than~$k(k+1)$ consecutive recursive calls in \cref{alg:3way}(c), that is, the recursion depth cannot increase by more than~$k(k+1)$ since the last time~$k$ decreased. We do this by showing in the following three lemmas that~$|N_G(T)|$ increases as consecutive recursive calls in \cref{alg:3way}(c) are made. Since the algorithm executes \cref{alg:greedy} if~$N_G(T) > k(k+1)$, this limits the number of consecutive recursive calls in \cref{alg:3way}(c).

\hui{The following lemma states that under certain conditions, the neighborhood of~$T$ does not decrease during the execution of a single recursive call.}

\begin{lemma} \label{lem:NT-nondecr}
 Let~$(G_0,k_0,F_0,T_0,w_0)$ be an \ELSS{} instance such that all leaves of \bmp{$G_0$} are contained in \bmp{$T_0$}. If the algorithm does not terminate before \cref{alg:3way}, then the instance~$(G',k',F',T',w')$ when executing \cref{alg:3way} satisfies~$|N_{G'}(T')| \geq \bmp{|N_{G_0}(T_0)|}$.

\end{lemma}
\begin{proof}

Since the algorithm does not terminate before \cref{alg:3way} it follows that \cref{alg:2nbrs,alg:greedy} are not executed, so consider the part of the algorithm before \cref{alg:2nbrs}.
\bmp{Throughout the proof we use~$(G,k,F,T,w)$ to refer to the instance at the time the algorithm evaluates it; initially~$(G,k,F,T,w) = (G_0,k_0,F_0,T_0,w_0)$, but actions such as contracting leaves may change the instance during the execution}. Suppose that all leaves of~$G$ are contained in~$T$. 
\bmp{We infer} that~$G[F]$ is acyclic, as otherwise \bmp{the algorithm} would return~$\emptyset$ \bmp{before reaching \cref{alg:3way}}. Removing the connected components of~$G$ that do not contain a vertex of~$F$ does not alter~$|N_G(T)|$. Afterwards we \bmp{know} that~$G$ is connected, as otherwise \bmp{the algorithm} would return~$\emptyset$. Consider a single degree-1 contraction step of a vertex~$v$ with~$F \neq \{v\}$ that results in the instance~$(G-v, k, F^*,T^*, w^*)$. Since we assume that all leaves are contained in~$T$, we have that~$v \in T$. Let~$u$ be the neighbor of~$v$. By \cref{def:contract} we \bmp{have}~$F^* =  (F \setminus \{v\}) \cup \{u\}$ and~$T^* = (T \setminus \{v\}) \cup \{u\}$. If~$u \in T$, then~$N_{G-v}(T^*) = N_G(T)$ and therefore their size is equal. If~$u \notin T$, then observe that~$u$ cannot be a leaf in~$G$ by assumption and therefore~$N_G(u) \setminus \{v\} \neq \emptyset$. Since~$T$ is connected and~$v$ is a leaf in~$T$ we get~$(N_G(u) \setminus \{v\}) \cap T = \emptyset$. It follows that~$|N_{G-v}(T^*)| \geq |N_G(T)|$. These arguments can be applied for each consecutive contraction step to infer~$|N_G(T)| \geq |N_{G_0}(T_0)|$ \bmp{for the instance~$(G,k,F,T,w)$ after all contractions.}

Next consider the step where if~$N_G(T)$ contains a vertex~$v \in F$, the vertex~$v$ is added to~$T$. Since~$G[F]$ is acyclic,~$G[T]$ is connected, and~$v \notin T$ is not a leaf, it follows that~$N_G(v) \setminus T \neq \emptyset$ and~$|N_G(v) \cap T| = 1$. It follows that~$|N_G(T \cup \{v\})| \geq \bmp{|N_G(T)|}$. Again these arguments can be applied iteratively. \bmp{For the instance~$(G,k,F,T,w)$ to which this step can no longer be applied,~$G[T]$ is a connected component of~$G[F]$.}

Next we get that~$N_G(F) \neq \emptyset$ as otherwise the algorithm would return~$\{(r,\emptyset)\}$ for some~$r \in F$. We also get~$k > 0$, as otherwise~$\emptyset$ would have been returned. 

Since none of the steps decreased the \bmp{size of the} neighborhood of~$T$, \bmp{for the instance~$(G',k',F',T',w')$ at the time \cref{alg:3way} is executed we conclude~$|N_{G'}(T')| \geq |N_{G_0}(T_0)|$ as required.}
\end{proof}

\hui{In the next lemma we show that the \bmp{size of the} neighborhood of~$T$ strictly increases as we make the recursive call in \cref{alg:3way}(c).}

\begin{lemma} \label{lem:NT-incr}
If the instance when executing \cref{alg:3way} is~$(G,k,F,T,w)$ and \cref{alg:3way}(c) branches on the instance~$(G,k,F \cup V(P),T \cup V(P),w)$, then either~$|N_G(T \cup V(P))| > |N_G(T)|$ or some vertex~$u \in N_G(T \cup V(P))$ is adjacent to at least two vertices in~$T \cup V(P)$.
\end{lemma}
\begin{proof}
Consider the path~$P = (v = v_1, \dots, v_\ell)$ with~$\degree_G(v_i) = 2$ for each~$1 \leq i < \ell$ and ($v_\ell \in N_G(T)$ or~$\degree_G(v_\ell) > 2)$ as defined in \cref{alg:3way}. The precondition of \cref{alg:3way}(c) gives that~$G[F \cup V(P)]$ is acyclic. Since~$T \subseteq F$, this implies that~$G[T \cup V(P)]$ is acyclic. It follows that~$V(P) \cap N_G(T) = \{v\}$ and~$\degree_G(v_\ell) > 2$.
Hence~$|N_G(T) \setminus \{v\}| = |N_G(T)|-1$ and~$|N_G(v_\ell) \setminus V(P)| \geq 2$.
Observe that~$N_G(T \cup V(P)) = (N_G(T) \setminus \{v\}) \cup (N_G(v_\ell) \setminus V(P))$ so if~$(N_G(T) \setminus \{v\}) \cap (N_G(v_\ell) \setminus V(P)) = \emptyset$ we have~\bmp{$|N_G(T \cup V(P))| > |N_G(T)|$}. Alternatively, suppose~$u \in (N_G(T) \setminus \{v\}) \cap (N_G(v_\ell) \setminus V(P))$. Then the second condition holds;~$u$ has at least one neighbor in~$T$ as~$u \in \bmp{N_G(T) \setminus \{v\}}$ and~$u$ is adjacent to~$v_\ell \notin T \setminus \{v\}$.
\end{proof}

Finally we combine \cref{lem:NT-nondecr,lem:NT-incr} to show~$|N_G(T)|$ is an upper bound to the number of consecutive recursive calls in \cref{alg:3way}(c).

\newcommand{\lemNTlarge}{If the recursion tree generated by the algorithm contains a path of~$i \geq 1$ consecutive recursive calls in \cref{alg:3way}(c), and~$(G,k,F,T,w)$ is the instance considered in \cref{alg:3way} where the $i$-th of these recursive calls is made, then~$|N_G(T)| \geq i$.}
\begin{lemma} \label{lem:NT-large}
 \lemNTlarge
\end{lemma}
\begin{proof}
 We use induction of~$i$. First suppose~\hui{$i=1$} and let~$(G,k,F,T,w)$ be the instance considered in \cref{alg:3way} where the \hui{first} of these recursive calls is made. If~$|N_G(T)| = 0$, then~$G[T]$ is a connected component of~$G$.
  \hui{However, since~$(G,k,F,T,w)$ is an instance considered in \cref{alg:3way} we know that \cref{prop:before-step1,prop:step1,prop:step3} apply. In particular~$N_G(F) \neq \emptyset$, ruling out that~$T=F$. However if}~$T \neq F$, then there are at least two connected components in~$G$ containing a vertex from~$F$, contradicting that~$G$ is connected (\cref{prop:before-step1}).
 By contradiction we can conclude that~$|N_G(T)| \geq 1 = i$.
 
 Suppose~\hui{$i \geq 2$} and let~$(G,k,F,T,w)$ be the instance considered in \cref{alg:3way} where the $i$-th recursive call is made. Let~$(G',k',F',T',w')$ be the instance considered in \cref{alg:3way} where the $(i-1)$-th recursive call is made. By induction we know~$|N_{G'}(T')| \geq i-1$. Let~$P$ be as in \cref{alg:3way} where the $(i-1)$-th recursive call is made, then by \cref{lem:NT-incr} we have that~$|N_{G'}(T' \cup V(P))| > |N_{G'}(T')|$ or some vertex~$u \in N_{G'}(T' \cup V(P))$ is adjacent to at least two vertices in~$T' \cup V(P)$. Since we know that the recursive call on~$(G',k',F' \cup V(P), T' \cup V(P), w')$ reaches \cref{alg:3way} with the instance~$(G,k,F,T,w)$, we can rule out that some vertex~$u \in N_{G'}(T' \cup V(P))$ is adjacent to at least two vertices in~$T' \cup V(P)$ as this would mean the recursive call ends in \cref{alg:2nbrs}. We can conclude instead that~$|N_{G'}(T' \cup V(P))| > |N_{G'}(T')|$.
 
 Note that since~$(G',k',F',T',w')$ is the instance in \cref{alg:3way} we have that \cref{prop:before-step1,prop:step1,prop:step3} apply. In particular,~$(G',k',F',T',w')$ is almost leafless, implying that all leaves in~$G'$ are contained in~$T'$. It follows that all leaves in~$G'$ are also contained in~$T' \cup V(P)$, so \cref{lem:NT-nondecr} applies to the input instance~$(G',k',F' \cup V(P), T' \cup V(P), w')$ (as recursively solved in \cref{alg:3way}) and the instance~$(G,k,F,T,w)$ (as considered in \cref{alg:3way} of that recursive call). So we obtain~$|N_G(T)| \geq |N_{G'}(T' \cup V(P))| > |N_{G'}(T')| \geq i-1$, that is,~$|N_G(T)| \geq i$.
\end{proof}

Since we know in \cref{alg:3way} that~$|N_G(T)| \leq k(k+1)$ (by \cref{prop:step3}) we can now claim that there are at most~$k(k+1)$ consecutive recursive calls of \cref{alg:3way}(c), leading to a bound on the recursion depth of~$\Oh(k^3)$. We argue that each recursive call takes~$\Oh(k n^3)$ time and since we branch at most three ways, we obtain a running time of~$3^{\Oh(k^3)} \cdot k n^3 = 3^{\Oh(k^3)} \cdot n^3$. However, with a more careful analysis we can give a better bound on the number of nodes in the recursion tree.

\newcommand{\lemRuntime}{The algorithm described in \cref{sec:alg} can be implemented to run in time~$2^{\Oh(k\log k)} \cdot n^3$.}
\begin{lemma} \label{lem:runtime}
 \lemRuntime
\end{lemma}
\begin{proof}
Consider the recursion tree of the algorithm. We first prove that each recursive call takes~$\Oh(k n^3)$ time \hui{(not including the time further recursive calls require)}. We then show that the recursion tree contains at most~$2^{\Oh(k \log k)}$ nodes.

\subsubsection*{Runtime per node}
Consider the input instance~$(G,k,F,T,w)$ with~$n = |V(G)|$ and~$m = |E(G)|$. We can verify that~$G[F]$ is acyclic in~$\Oh(|F|)$ time using DFS. Again using DFS, in~$\Oh(n + m)$ time identify all connected components of~$G$ and determine whether they contain a vertex of~$F$. We can then in linear time remove all connected components that contain no vertex of~$F$ and return~$\emptyset$ if more than one component remains. Finally exhaustively contracting degree-1 vertices into \bmp{their} neighbor is known to take~$\Oh(n)$ time. Updating~$F$ and~$T$ only results in~$\Oh(1)$ overhead for each contraction. \hui{Exhaustively adding vertices~$v \in N_G(T) \cap F$ to~$T$ can be done in~$\Oh(n)$ time since it corresponds to finding a connected component in~$G[F]$ which is acyclic.}

For \cref{alg:2nbrs} we can find a vertex~$v \in N_G(T)$ with two neighbors in~$T$ in~$\Oh(n^2)$ time by iterating over all neighbors of each vertex in~$T$. 


Determining the size of the neighborhood in \cref{alg:greedy} can be done in~$\Oh(n^2)$ time. 
Applying \cref{lem:greedy:naive} takes~$\Oh(k n^3)$ time. So excluding the recursive call, \cref{alg:greedy} can be completed in~$\Oh(k n^3)$ time.

For \cref{alg:3way} an arbitrary~$v \in N_G(T)$ can be selected in~$\Oh(1)$ time, and the path~$P$ can be found in~$\Oh(|P|) = \Oh(n)$ time as described in \cref{footnote}. Finally the results of the three recursive calls in \cref{alg:3way} are combined. Selecting an arbitrary tree from~$\T_G(\kX'_i)$ for any~$i \in [3]$ involves selecting and arbitrary description~$(r,\X) \in \kX'_i$ and then selecting, for each~$X \in \X$ and arbitrary vertex~$v \in X$. 
Now the tree can be found using DFS starting from~$r$ exploring an acyclic graph until it reaches the selected vertices from a set~$X \in \X$. This all takes~$\Oh(n)$ time. The weights of the selected secluded trees can be found in~$\Oh(n)$ time as well. Finally we take the union of (a selection of) the three sets of descriptions. Since these sets are guaranteed to be disjoint, this can be done in constant time.

\subsubsection{Number of nodes}
We now calculate the number of nodes in the recursion tree. To do this, label each edge in the recursion tree with a label from the set~$\{\mathsf{1}, \mathsf{2}, \mathsf{3a}, \mathsf{3b}, \mathsf{3c}\}$ indicating where in the algorithm the recursive call took place. Now observe that each node in the recursion tree can be uniquely identified by a sequence of edge-labels corresponding to the path from the root of the tree to the relevant node. We call such a sequence of labels a \emph{\strace}.

Note that for all recursive calls made in~$\mathsf{1}, \mathsf{2}, \mathsf{3a},$ and~$\mathsf{3b}$ the parameter ($k$) decreases, and for the call made in~$\mathsf{3c}$ the parameter remains the same. If~$k \leq 0$ we do not make further recursive calls, so the \strace contains at most~$k$ occurrences of~$\mathsf{1}, \mathsf{2}, \mathsf{3a},$ and~$\mathsf{3b}$. Next, we argue there are at most~$k(k+1)$ consecutive occurrences of~$\mathsf{3c}$ in the \strace. 

Suppose for the sake of contradiction that the \strace contains~$k(k+1)+1$ consecutive occurrences of~$\mathsf{3c}$. Let~$(G,k,F,T,w)$ be the instance considered in \cref{alg:3way} where the last of these recursive calls is made. By \cref{lem:NT-large} we have~$|N_G(T)| > k(k+1)$. 
This contradicts \cref{prop:step3}, so we can conclude the \strace contains at most~$k(k+1)$ consecutive occurrences of~$\mathsf{3c}$ and hence any valid \strace has a total length of at most~$k \cdot k(k+1) = \Oh(k^3)$.

In order to count the number of nodes in the recursion tree, it suffices to count the number of different valid \straces. Since a \strace contains at most~$k$ occurrences that are not~$\mathsf{3c}$ we have that the total number of \straces of length~$\ell$ is~$\binom{\ell}{k} \cdot 4^k \leq \ell^k \cdot 4^k = (4\ell)^k$.
We derive the following bound on the total number of valid \straces using the fact that~$(k^c)^k = (2^{\log (k^c)})^k = 2^{\Oh(k \log k)}$:
\begingroup\allowdisplaybreaks
\begin{align*}
    \sum_{1 \leq \ell \leq k^2(k+1)} (4\ell)^k
    &\leq k^2(k+1) \cdot (4k^2(k+1))^k 
       = 2^{\Oh(k \log k)}
    .
\end{align*}
\endgroup

We can conclude that the total number of nodes in the recursion tree is at most~$2^{\Oh(k \log k)}$ so the overall running time is~$2^{\Oh(k \log k)} \cdot k n^3 = 2^{\Oh(k \log k)} \cdot n^3$.
%
\end{proof}

\subsection{Counting, enumerating, and finding large secluded trees}\label{sec:enum_count}

With the algorithm of \cref{sec:alg} at hand we argue that we are able to enumerate $k$-secluded trees, count such trees containing a specified vertex, and solve \LST{}.

\newcommand{\thmEnumeration}{There is an algorithm that, given a graph~$G$, weight function~$w$, and integer~$k$, runs in time~$2^{\Oh(k \log k)} n^4$ and outputs a set of descriptions~$\kX$ such that~$\T_G(\kX)$ is exactly the set of maximum-weight $k$-secluded trees in~$G$. Each such tree~$H$ is described by~$|V(H)|$ distinct descriptions in~$\kX$.}
\begin{theorem}\label{thm:enumeration}
 \thmEnumeration
\end{theorem}
\begin{proof}
Given the input~$(G,k,w)$, we proceed as follows. For each~$v \in V(G)$, let~$\kX_v$ be the output of the \ELSS{} instance~$(G,k,F=\{v\},T=\{v\},w)$ and let~$w_v$ be the weight of an arbitrary $k$-secluded supertree in~$\T_G(\kX_v)$, or~$0$ if~$\kX_v = \emptyset$. Note that all $k$-secluded trees described by~$\kX_v$ have weight exactly~$w_v$. Let~$w^* := \max_{v \in V(G)} w_v$. If~$w^* = 0$ then there are no $k$-secluded trees in~$G$ and we output~$\kX = \emptyset$; otherwise we output~$\kX := \bigcup \{ \kX_v \mid v \in V(G) \wedge w_v = w^* \}$.

Clearly~$\T_G(\kX)$ is the set of all $k$-secluded trees in~$G$ of maximum weight. 
Since each~$\kX_v$ is non-redundant, each maximum-weight $k$-secluded tree~$H$ is described by exactly~$|V(H)|$ descriptions in~$\kX$.
\end{proof}

By returning an arbitrary maximum-weight $k$-secluded tree described by any description in the output of \cref{thm:enumeration}, we have the following consequence.

\begin{corollary}
There is an algorithm that, given a graph~$G$, weight function~$w$, and integer~$k$, runs in time~$2^{\Oh(k \log k)} n^4$ and outputs a maximum-weight $k$-secluded tree in~$G$ if one exists.
\end{corollary}

The following theorem captures the consequences for counting.

\newcommand{\thmCounting}{There is an algorithm that, given a graph~$G$, vertex~$v \in V(G)$, weight function~$w$, and integer~$k$, runs in time~$2^{\Oh(k \log k)} n^3$ and counts the number of $k$-secluded trees in~$G$ that contain~$v$ and have maximum weight out of all $k$-secluded trees containing~$v$.}
\begin{theorem} \label{thm:counting}
 \thmCounting
\end{theorem}
\begin{proof}
Construct the \ELSS{} instance~$(G,k,F=\{v\},T=\{v\},w)$ and let~$\kX$ be the output obtained by the algorithm described in \cref{sec:alg}. Note that this takes~$2^{\Oh(k \log k)} n^3$ time by \cref{lem:runtime}. Since the definition of \ELSS{} guarantees that~$\kX$ is non-redundant, each maximum-weight tree containing~$v$ is described by exactly one description in~$\kX$. To solve the counting problem it therefore suffices to count how many distinct $k$-secluded trees are described by each description in~$\kX$.

By \cref{def:description}, for each description~$(r, \mathcal{X}) \in \kX$, each way of choosing one vertex from each set~$X \in \mathcal{X}$ yields a unique $k$-secluded tree. Hence the total number of maximum-weight $k$-secluded trees containing~$v$ is:
\begin{equation*}
  \sum_{(r,\X) \in \kX} \prod_{X \in \X} |X|,
\end{equation*}
which can easily be computed in the stated time bound.
\end{proof}

\section{Conclusion}\label{sec:conclusion}

We revisited the \textsc{$k$-Secluded Tree} problem first studied by Golovach et al.~\cite{DBLP:journals/jcss/GolovachHLM20}, leading to improved FPT algorithms with the additional ability to count and enumerate solutions. The non-trivial progress measure of our branching algorithm is based on a structural insight that allows a vertex that belongs to the \emph{neighborhood} of every solution subtree to be identified, once the solution under construction has a sufficiently large open neighborhood. As stated, the correctness of this step crucially relies on the requirement that solution subgraphs are acyclic. It would be interesting to determine whether similar branching strategies can be developed to solve the more general \textsc{$k$-Secluded Connected $\mathcal{F}$-Minor-Free Subgraph} problem; the setting studied here corresponds to~$\mathcal{F} = \{K_3\}$. While any $\mathcal{F}$-minor-free graph is known to be sparse, it may still contain large numbers of internally vertex-disjoint paths between specific pairs of vertices, which stands in the way of a direct extension of our techniques.

A second open problem concerns the optimal parameter dependence for \textsc{$k$-Secluded Tree}. The parameter dependence of our algorithm is~$2^{\Oh(k \log k)}$. Can it be improved to single-exponential, or shown to be optimal under the Exponential Time Hypothesis?

\bibliographystyle{plainurl}
\bibliography{refs}
\end{document}